\UseRawInputEncoding
\pdfoutput=1
\documentclass[10pt,pra,aps,twocolumn]{revtex4-1}
\usepackage{amsmath,bbm}
\usepackage{amsthm}
\usepackage{latexsym}
\usepackage{amssymb}
\usepackage{float}
\usepackage{graphicx}           
\usepackage{color}
\usepackage{mathpazo}
\usepackage{relsize}
\usepackage{braket}
\usepackage[colorlinks=true,linkcolor=blue,citecolor=magenta,urlcolor=blue]{hyperref}
\usepackage{changes}
\usepackage{cancel}
\usepackage{xcolor}

\usepackage[utf8]{inputenc}
\usepackage[T1]{fontenc}

\newcommand{\be}{\begin{equation}} 
\newcommand{\ee}{\end{equation}}
\newcommand{\beq}{\begin{eqnarray}}
\newcommand{\eeq}{\end{eqnarray}}

\def\squareforqed{\hbox{\rlap{$\sqcap$}$\sqcup$}}
\def\qed{\ifmmode\squareforqed\else{\unskip\nobreak\hfil
\penalty50\hskip1em\null\nobreak\hfil\squareforqed
\parfillskip=0pt\finalhyphendemerits=0\endgraf}\fi}
\def\endenv{\ifmmode\;\else{\unskip\nobreak\hfil
\penalty50\hskip1em\null\nobreak\hfil\;
\parfillskip=0pt\finalhyphendemerits=0\endgraf}\fi}

\newcommand{\I}{\mathbbm{1}}

\newcommand{\ra}{\rangle}
\newcommand{\la}{\langle}

\makeatletter
\newtheorem*{rep@theorem}{\rep@title}
\newcommand{\newreptheorem}[2]{%
\newenvironment{rep#1}[1]{%
 \def\rep@title{#2 \ref{##1}}%
 \begin{rep@theorem}}%
 {\end{rep@theorem}}}
\makeatother

\def\tr{\mbox{tr}}
\newtheorem{thm}{Theorem}
\newreptheorem{thm}{Theorem}
\newtheorem{lemma}{Lemma}

\newtheorem{defi}{Definition}

\begin{document}

\title{Certification of multi-qubit quantum systems with temporal inequalities}

\author{Gautam Sharma}
\author{Chellasamy Jebarathinam}
\author{Sk Sazim}
\author{Remigiusz Augusiak}
\affiliation{Center for Theoretical Physics, Polish Academy of Sciences, Aleja Lotnik\'{o}w 32/46, 02-668 Warsaw, Poland}

\begin{abstract}
Demonstrating contextual correlations in quantum theory through the violation of a non-contextuality inequality necessarily needs some ``contexts" and thus assumes some compatibility relations between the measurements. As a result, any self-testing protocol based on the maximal violation of such inequality is not free from such assumptions. In this work, 
we propose temporal inequalities derived from non-contextuality inequalities for multi-qubit systems without assuming any compatibility relations among the measurements. We demonstrate that the new inequalities can be maximally violated via a sequential measurement scenario. Moreover, using the maximal violation of these temporal inequalities we are able to certify multi-qubit graph states and the measurements.
\end{abstract}

\maketitle

\section{Introduction}
Multi-qubit entangled quantum systems are a ubiquitous resource in a large variety of quantum information processing tasks like quantum computation \cite{RB01,KMN+07}, quantum error correction \cite{PhysRevA.52.R2493,RevModPhys.87.307}, quantum communications \cite{PhysRevX.10.021071,Hilaire2021resource}, quantum simulations \cite{lanyon2010towards,ma2011quantum} and cryptographic protocols \cite{PhysRevLett.98.020503,Epping_2017}. It is desirable to make sure that the quantum systems employed in these tasks perform as specified by the provider. There exist a number of ways to do such verification of the devices (quantum systems). The standard tomography-based methods used to certify quantum devices are constrained by two key limitations: a) the substantial overhead of resources required during the process, and b) their infeasibility when applied to larger system sizes. \cite{PhysRevA.66.012303,PhysRevLett.95.210502}. Fortunately, there exist alternate certification techniques, namely, self-testing, which are not only feasible for larger system sizes but also require fewer resources.

Self-testing of quantum states and measurements is a remarkable way of certifying the quantum devices \cite{MY04}, in terms of reducing the required resources. A plethora of self-testing schemes have been introduced using non-local correlations based on the maximal quantum violation of Bell-type inequalities \cite{Supic2020selftestingof,PhysRevLett.117.070402,PhysRevLett.124.020402,panwar2022elegant}. Although the self-testing schemes based on Bell-type inequalities provide a device-independent characterization of quantum systems, these are restricted by the fact that they require spatial separation between the subsystems. Thus, there is a hidden assumption of using compatible measurement devices as the spatially separated measurements are compatible with each other. 

Recently, various certification schemes based on Kochen-Specker contextuality \cite{KS67} and temporal correlations \cite{LG85,PhysRevLett.113.050401,PhysRevLett.115.120404} have also been proposed for certifying quantum devices \cite{PhysRevLett.122.250403,IMO+20,SSA20,MMJ+21,SJA22}. Demonstrating quantum contextuality and temporal correlations do not require any spatial separation and can be observed by using sequential quantum measurements on a single quantum system. Thus, we can get rid of the limitation of spatial separation between subsystems. However, the certification schemes based on the violation of non-contextuality inequalities also assume certain compatibility relations between the measurement devices \cite{BRV+19,IMO+20,SSA20} whereas the certification schemes using temporal inequalities assume compatibility relations on measurements as well as a) the preparation device always
prepares a maximally mixed state and b) the measurement device always returns the post-measurement state and does not have any memory \cite{SSA20,DMS+22}. Recently, we propose a self-testing scheme for two-qubit devices based on sequential correlations where we do not require the assumptions of compatibility on the measurements \cite{twoqubit_temp}.  

In this work, we extend the certification scheme of Ref. \cite{twoqubit_temp} to multi-qubit systems without assuming that the measurement operators are commuting. Specifically, we extend our previous work Ref. \cite{SJA22}, in which we introduced a family of temporal inequalities for the certification of multi-qubit quantum systems. To these non-contextual inequalities, by dropping the assumption of compatibility relations between the measurement operators, we modify the non-contextuality inequality by using sequential correlations such that the modified inequalities are also maximally violated by the multi-qubit graph states. In the self-testing protocol, we are able to certify the previously ``assumed" compatibility relations between the measurements from the maximal violation of the inequality. In what follows, we are able to certify all the pairs of anti-commuting operators as well as the multi-qubit graph states.

This paper is organized as follows. In Sec \ref{preliminaries}, we give a brief description of the contextuality scenario and its temporal extension, the graph states, and the self-testing statement. Next, in Sec. \ref{simplestinequality} we introduce the modified non-contextuality inequality for three-qubits and its further modification using sequential correlations and then demonstrate self-testing of three-qubit graph states. In Sec \ref{scalable}, we extend the same procedure for the $n$-qubit graph states. 

\section{Preliminaries}\label{preliminaries}
Imagine a sequential measurements scenario that enables us to study both contextual and temporal correlations as we will describe later. In order to perform sequential measurements we assume that these measurements do not physically destroy the system. Let us consider a set of measurements $\{M_i\}$ that are measured in a sequence on a physical system. The possible measurement outcomes $m_i$ can take the values $\pm1$. For a sequence of $n$ measurements performed in the order $M_1 \rightarrow M_2\rightarrow\ldots \rightarrow M_n$, the sequential correlations are defined as follows. 
\begin{align}\label{seq:exp}
   \la M_1\ldots M_n\ra_{seq}=\sum_{m_i=\pm1}m_1 \ldots m_n p(m_1, \cdots ,m_n|\{M_i\}),
\end{align}
where the $p(m_1, \cdots, m_n|\{M_i\})$ is the joint probability distribution of measurement outcomes for measurements performed in the aforementioned order. 

Since, we do not assume that the measurements performed in a sequence commute, we cannot exploit the fact that 
the dimension of the underlying Hilbert space is arbitrary to 
invoke the Naimark dilation technique to 
represent the observed correlations obtained with arbitrary 
measurements (also POVM) in terms of projective measurements.
Therefore, we need to assume the measurements $M_i$ to be projective. Also,
from now on, slightly abusing the notation, by $M_i$ we denote quantum observables representing those measurements; their eigenvalues are
$\pm1$, so that $M_i^2=\mathbbm{1}$.
%
%

Let us now notice that the expectation values \eqref{seq:exp} for 
for the two (noncummiting) measurements performed in a sequence $M_1\rightarrow M_2$, expresses in terms of the Born's rule as
\cite{MHL21, PhysRevA.98.062115}
\begin{align*}
 \la M_1 M_2 \ra _{seq} &:= \frac{1}{2}(\la M_1M_2\ra+\la M_2 M_1\ra)\nonumber\\
 &=  \frac{1}{2}\tr\left(\rho  \{M_1,M_2 \} \right),
\end{align*}
where $\{A,B \}=AB+BA$.
Similarly, for three measurements $M_1\rightarrow M_2 \rightarrow M_3$, the third order sequential correlations take the following form 
\begin{eqnarray*}
 \la M_1 M_2 M_3\ra _{seq}&:=& \la M_1M_2M_3\ra +\la M_1M_3M_2\ra\nonumber \\&&+\la M_2M_3M_1+\la M_3M_2M_1\ra \nonumber\\ &=&\frac{1}{4}\tr\left(\rho\{M_1,\{M_2,M_3\}\}\right).
\end{eqnarray*}
Therefore, the sequential correlations for $n$ measurements, $\la M_1 \ldots M_n\ra_{seq}$, can exactly be expressed as below 
\begin{align}\label{nth}
    \frac{1}{2^{n-1}}\tr\left(\rho \{M_1,\{M_2,\ldots,\{M_{n-1},M_n\} \}\}\right).
\end{align}
For a complete analysis of $n$-th order sequential correlations, we refer our enthusiastic reader to the appendix of Ref. \cite{PhysRevA.98.062115}.
\subsection{Non-contextuality inequalities in sequential  measurement scenario}
\label{cont_scenario}
Consider now, that the measurements $\{M_l\}$ are such that we know the compatibility relations among them. This knowledge of compatibility relations allows us to identify the contexts, where all the measurements belonging to a ``context" are compatible with each other. An experiment, that realizes the sequential measurement of $n$ measurements $\{M_{l}\}$ belonging to a ``context", generates a joint probability distribution function that is independent of the order in which the measurements are performed. Since the joint probability distribution functions are now independent of the order, the sequential correlation function reduces to the expected values as $\la M_{1}\ldots M_{n}\ra_{seq} =\la M_{1}\ldots M_{n}\ra$.

These expectation values can then be used to construct a non-contextuality inequality. The non-contextuality inequalities are defined as a linear function of sequential correlations as following
\beq \label{gennoncineq}
\mathcal{I}_{NC}=\sum c_{i,j,k,\cdots }\la M_i  M_j M_k\ldots\ra \le \eta_C \le \eta_Q
\eeq
where $c$'s are real coefficients, and $\eta_C$ and $\eta_Q$ are the classical and quantum bounds. Whenever $\eta_Q>\eta_C$ is observed it leads to the refutation of the non-contextuality assumption in quantum theory \cite{KS67}.

The classical bound $\eta_C$ is obtained if a non-contextual hidden variable model can describe the joint probability distributions of the measurement outcomes \cite{KS67}.
In a non-contextual model, the expectation values can be factorized and each expectation value is deterministic, i.e.,
$\la M_{1} \ldots M_{n} \ra=\prod_l m_{l}$, with $m_{l}=\pm 1$.  Whereas the quantum bound $\eta_Q$ is obtained by maximizing over the set of all quantum states and measurements in any Hilbert space, and we can write 
\begin{align*}
    \eta_Q=\sup_{\{M_{l}\},\rho}\left[\sum c_{i,j,k,\cdots}\tr({\rho }\ M_i M_j M_k \ldots)\right],
\end{align*}
where $M_{l}$ are Hermitian operators acting on the Hilbert space $\mathcal{H}$ with eigenvalues $\pm1$ and $\rho=\ket{\psi}\bra{\psi}$. 
It has been demonstrated that the maximal violation of the non-contextuality inequalities can be used to certify multi-qubit entangled states and the measurements \cite{IMO+20,SJA22}. However, we would like to get rid of the assumption of having information about the compatibility of measurements. So we will prescribe a method to construct temporal inequalities from the non-contextuality inequalities in the next section.
\subsection{Temporal inequalities in sequential measurement scenario}
\label{temp_cont_scenario}
 %
%
Let us now generalize the above findings to the case of the complete graph state of any number of qubits. To this end, we will construct temporal inequalities in the sequential measurement scenario by modifying the non-contextuality inequalities in Eq. \eqref{gennoncineq}. To do so, we replace expectation value terms of compatible measurements appearing in \eqref{gennoncineq}, with one or more sequential correlations of $\{M_i\}$ such that there is at least one ``expectation" value term for all permutations of a sequence of the measurements. We explain the methodology with a few examples below.  

\textit{Generating all permutations of expectation values}.-- 
Let us consider the simplest example consisting of two compatible measurements $M_{1}, M_{2}$ and we will make the following replacement with sequential correlation function of $M_1 \rightarrow M_2$
\begin{align*}
    \la M_1M_2 \ra\rightarrow \la M_1M_2\ra_{seq}=\frac{1}{2}(\braket{M_1M_2}+\braket{M_2M_1}),
\end{align*}
As there are only two possible permutations, we get both of them with one sequential correlation. Next, we consider the expectation value of three compatible measurements and replace it with a pair of sequential correlations of the types $M_1 \rightarrow M_2 \rightarrow M_3$ and $M_2 \rightarrow M_1 \rightarrow M_3$, i.e., we have 
\begin{align}
 \la M_1M_2M_3\ra_{\pi} =\frac{1}{2}\left(\la M_1M_2M_3\ra_{seq}+\la M_2M_1M_3\ra_{seq}\right) \nonumber,
\end{align}
where  $\la M_1M_2M_3\ra_{\pi}$ contains sequential correlations such that all possible permutations of expectation values are included at least once. It can be readily checked that the $\la M_1M_2M_3\ra_{\pi}$ contains all possible six permutations of expectation values with $M_1,M_2,M_3$. The factor of $1/2$ is chosen so that the maximum value of the $\la M_1M_2M_3\ra_{\pi}$ is one. 
In general, we can do this for an arbitrary number of $n$ measurements which will require us to replace expectation value terms by a set of sequential correlations and divide the sum by the number of sequential correlations. A minimum of $n!/2^{n-1}$ number of sequential correlations are required to generate all permutations of expectation values. This is due to the fact that there are $n!$ possible permutations and each sequential correlation gives $2^{n-1}$ permutations of expectation values. Therefore, we construct a temporal inequality from Eq. \eqref{gennoncineq} as following
\begin{align}\label{gentempineq}
  \mathcal{I}_T=  \sum c_{\{M_i\}}\braket{M_1\ldots M_n}_{\pi}\leq \eta_C \leq \eta_Q, 
\end{align}
where $\la M_1\ldots M_n\ra_{\pi}$ contains all possible sequential correlations of $M_1,\ldots,M_n$ such that all $n!$ permutations of expectation values appear at least once. 

Now, we can derive the classical bound of \eqref{gentempineq} using any deterministic strategy, which is due to the fact that similar to Bell-type and non-contextuality scenarios, the temporal inequalities also rely on the assumption of the existence of joint probability distribution for the measurement outcomes \cite{MKT+14}. Whereas the quantum bound can be obtained by maximizing over the set of all pure quantum states and projective measurements, as argued earlier.
\subsection{Graph states}
The goal of this work is to introduce a method of certification of multi-qubit genuinely entangled GHZ states that are a particular instance of a larger class of states called graph states, corresponding to complete graphs.

Before moving to our results let us define the graph states (see, e.g., Ref. \cite{HDE+06}). Consider a graph $\mathcal{G}=(\mathcal{V},\mathcal{E})$ consisting of $n$ vertices from a set $\mathcal{V}$ and a set of edges $\mathcal{E}$ that connect some pairs of vertices; in particular, a graph is termed complete when for any pair of vertices there exists an edge connecting them. Now, in order to associate an $n$-qubit state to any connected graph we follow the stabilizer formalism \cite{Got96} and to each vertex $v_i\in \mathcal{V}$ we associate an operator defined as 
\begin{equation} \label{Gi}
    G_{i} = X_{i}\otimes  \bigotimes_{j \in \mathcal{N}(i)} Z_{j}, 
\end{equation} 
where $X_i$ and $Z_i$ are the qubit Pauli operators. The operators $X_i$ act on the vertex $v_i$, while the operators $Z_j$ act on the vertices $v_j\in \mathcal{N}_i$. In this way we arrive at the following definition.
\begin{defi}\label{def:gs}
A graph state $|G \rangle $ associated with the graph $\mathcal{G}=(\mathcal{V},\mathcal{E})$ is defined as the unique state stabilized by the corresponding operators $G_i$ \eqref{Gi}, that is,
\begin{equation}\label{eq2}
   G_{i}|G \rangle = |G \rangle ,\qquad  \forall i=1,\dots,n.
\end{equation}
In other words, $\ket{G}$ is the unique common eigenvector of all stabilizers $\{G_i\}$ with eigenvalue $+1$.
\end{defi}
%

In this work, we consider the $n$-qubit graph states which are constructed using a complete graph, i.e., the graphs in which all pairs of vertices are connected with one another. The stabilizing operators for such connected graphs have the following form
\begin{equation}\label{ngraph}
    G_i=Z_1\ldots Z_{i-1}X_iZ_{i+1}\ldots Z_n,
\end{equation}
with $i=1,\ldots,n$. They stabilize an $n$-qubit graph state 
that is local-unitary equivalent to the GHZ state $(1/\sqrt{2})(\ket{0}^{\otimes n}+\ket{1}^{\otimes n})$.
The simplest example of such a graph state is associated with a complete three vertex graph (left of Fig. \ref{graphs1and2}) state which is stabilized by the following three stabilizing operators:
\begin{eqnarray}\label{Stab3}
G_1 &=& X \otimes Z \otimes Z, \nonumber \\
G_2 &=& Z \otimes X \otimes Z, \nonumber \\
G_3 &=& Z \otimes Z \otimes X,
\end{eqnarray}
and can be stated as
\begin{eqnarray}\label{FCg}
|G_3 \rangle  &=& \frac{1}{\sqrt{8}}( 
{\left| 000 \right\rangle} 
+ {\left| 100 \right\rangle} 
+ {\left| 010 \right\rangle} 
- {\left| 110 \right\rangle} \nonumber \\
&&\hspace{0.75cm}+ {\left| 001 \right\rangle} 
- {\left| 101 \right\rangle} 
- {\left| 011 \right\rangle} 
- {\left| 111 \right\rangle} 
).
\end{eqnarray}
The next simplest set of stabilizers of this form is associated with a connected four vertex graph (on the right of Fig. \ref{graphs1and2}) and in a similar way we can construct further such $n$-qubit connected graph states. 
%

\begin{figure}
    \centering
    \includegraphics[scale=0.45]{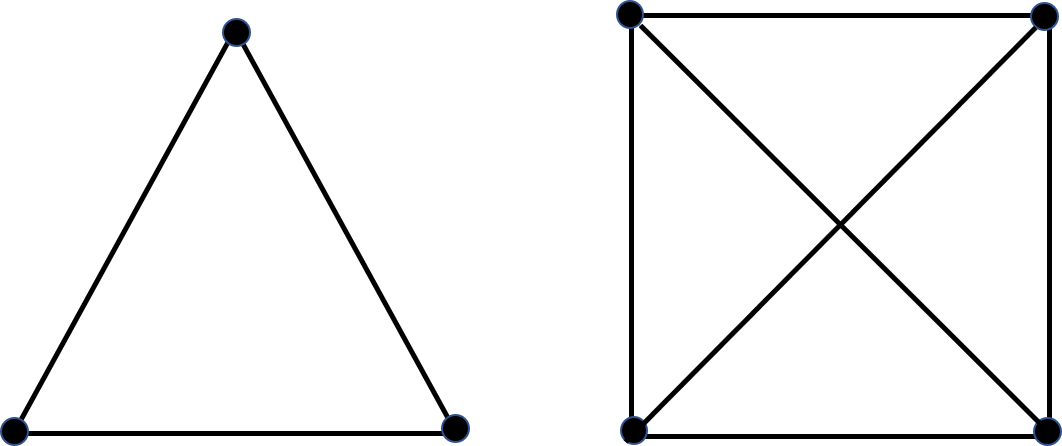}
    \caption{Two simplest connected graphs with 3 and 4 vertices respectively.} \label{graphs1and2}
  \end{figure}
\subsection{Self-testing with temporal inequalities}
In its original formulation put forward in Ref. \cite{MY04} (see also Ref. \cite{Supic2020selftestingof}), self-testing aims to certify an unknown quantum state and a set of measurements from the statistics obtained in an experiment, up to certain equivalences such as local isometries and additional of some extra degrees of freedom \cite{MY04}. The self-testing based on Bell inequality violations is by definition device-independent as it does not, in general, require any assumption about the state and the measurements. However, in the case of contextuality-based certification one has to make an additional assumption that the measurements involved obey certain compatibility relations \cite{BRV+19}, which in the case of Bell scenario is automatically satisfied due to the spatial 
separation between observers. In general it is dificult to certify in an experiment whether such compatibility relations are obeyed by the measurements. There exist other self-testing schemes based on the violation of non-contextual inequalities which do not require making this assumption \cite{SSA20,DMS+22}.


Our aim here is to introduce the temporal-correlations version of the contextuality-based certification schemes introduced recently in Ref. \cite{SJA22}, which does not rely on any assumption about the compatility of
the involved measurements. To this aim, let us consider a temporal inequality, say $\mathcal{I}_T$ constructed from a non-contextuality inequality, $\mathcal{I}_{NC}$, i.e., $\mathcal{I}_{NC}\rightarrow \mathcal{I}_T$ using the techniques in \ref{temp_cont_scenario}. We have an experiment that realizes the sequential measurement scenario for $\mathcal{I}_T$, such that both the state (in general mixed) and measurements (in general non-projective) is unknown. 

We then consider a reference experiment with a known pure state $\ket{\tilde{\psi}}\in\mathbbm{C}^d$ for some $d$ and known measurements $\tilde{M}_i$ acting on $\mathbbm{C}^d$.  We assume that our reference experiment also achieves the maximal quantum value of the given expression $\mathcal{I}_{T}$. With these two experiments at hand, we can now introduce
the definition of self-testing in the case of contextuality or sequential scenarios that was put forward in \cite{IMO+20}:
\begin{defi}\label{defselftesting}
Suppose an unknown  state $|\psi \ra  \in \mathcal{H}$ and a set of measurements $ M_i $  violate the temporal inequality $\mathcal{I}_T$  maximally. Then, this maximal quantum violation self-tests the state $| \tilde{\psi}  \ra  \in \mathbbm{C}^d$ and the set of measurements $ \tilde{M_i} $ such that there exists exists a projection $P:\mathcal{H} \rightarrow \mathbb{C}^d$ and a unitary $U$ acting on $\mathbbm{C}^d$ so that
\beq\label{DefSelf}
U^\dagger (P\, M_i\, P^\dagger) U = \tilde{M}_i, \:\:{\rm and}\:\:
U (P\,|\psi \ra) =|\tilde{\psi} \ra.
\eeq
\end{defi}
Speaking alternatively, the above definition says
that based on the observed non-classicality one is able to identify 
a subspace $V=\mathbbm{C}^d$ in $\mathcal{H}$ on which 
all the observables act invariantly. Equivalently, 
$M_i$ can be decomposed as $M_i=\hat{M}_i\oplus \bar M_{i}$, where
$\hat{M}_i$ acts on $V$, whereas $\bar M_i$ acts on the orthogonal complement of $V$ in $\mathcal{H}$; in particular, $\bar M_i\ket{\psi}=0$. Moreover, there exists a unitary $U^{\dagger}\,\hat{M}_i\,U=\tilde{M}_i$.
\section{Self-testing of three-qubit complete graph state}\label{simplestinequality}

\subsection{Non-contextuality inequality}
Taking inspiration from the previous works \cite{IMO+20,twoqubit_temp,SJA22,DMS+22}, let us now introduce the following non-contextuality inequality for self-testing the complete graph state of three qubits and a set of nine observables:
\begin{align}\label{nc_new}
\mathcal{I}_3=&\braket{A_1A_2B_3M_{12}}+\braket{A_1B_2A_3M_{31}}+\braket{B_1A_2A_3M_{23}} \nonumber \\
&+\braket{A_1B_2B_3}+\braket{B_1A_2B_3}+\braket{B_1B_2A_3}-\braket{A_1A_2A_3}  \nonumber \\
&+\braket{M_{23}M_{13}}+\braket{M_{13}M_{12}}+\braket{M_{23}M_{12}} \nonumber \\ &\le \eta_C=8\leq\eta_Q=10,
\end{align}
where $A_i$, $B_j$ and $M_{ij}$ are dichotomic observables with eigenvalues $\pm1$ and hence satisfy $A^2_i=B^2_i=M_{ij}^2=\I$. Note that the notation for $M_{ij}$ is such that $M_{ij}$ and $M_{ji}$ represent the same observable. This enables us to identify the symmetries of the inequality which are convenient for the later proofs. 

It can be checked that the classical bound of the inequality \eqref{nc_new} can be obtained by assigning the values $\pm1$ to each variable $A_i$, $B_j$ and $M_{ij}$ which implies $\eta_C = 8$. And, the maximum algebraic value of $\mathcal{I}_3$ is $\eta_Q=10$, which we can achieve with the following choice of quantum operators:
\begin{align} \label{ABCs}
{A}_i = X_i , \qquad
{B}_j = Z_j , \qquad M_{ij}=X_i\otimes X_j \otimes Z_k,
\end{align}
and the graph state $|G_3 \ra$ in Eq. (\ref{FCg}), where $i\neq j\neq k \in\{1,2,3\}$. In fact, to derive the inequality \eqref{nc_new} we exploited the 
fact that the graph state $\ket{G_3}$ is an eigenstate of each of 
the product of observables appearing in $\mathcal{I}_3$, except for
$A_1A_2A_3$ for which it is an eigenvector with eigenvalue $-1$.

Notice also that for readability, we use a compact notation to represent the observables, eg., $X_i$, and $Z_i$ represent respectively the Pauli $X$ and $Z$ matrices acting on $i$-th qubit. As illustrations, for $n=3$ the following should read like, ${A}_2 = \mathbbm{1} \otimes X \otimes \mathbbm{1}$ and $M_{13}=M_{31}=X\otimes Z \otimes X$. We note that these operators obey the following commutation relations for $i\neq j \neq k \in \{1,2,3\}$,
\begin{align}\label{comm_3}
    [{A}_i,{A}_j]&=[{A}_i,{B}_j]=[{B}_i,{B}_j]=[{A}_i,M_{ij}]=[{A}_j,M_{ij}] \nonumber\\ 
    &=[{B}_k,M_{ij}]=[M_{ij},M_{kj}]=0,
\end{align}
as well as the following anti-commutation relations 
\begin{align}\label{anticomm_3}
    \{A_i,B_i\}=\{A_k,M_{ij}\}=\{B_{i(j)},M_{ij}\}=0.
\end{align}
We also notice that the non-contextuality inequality \eqref{nc_new}, respects certain symmetries. These symmetries can be listed as follows 
\begin{align}
   & 1.\:\: \text{Permutation of indices} \quad i \leftrightarrow j. \label{symm1}\\ 
   & 2.\:\: \text{Cyclic permutation of indices, induced by \eqref{symm1}.} \nonumber
\end{align}

In the next subsection, we propose a temporal non-contextuality inequality derived from the Eq. \eqref{nc_new} using the techniques proposed in \ref{temp_cont_scenario}. We will not assume the commutation relations from Eq. \eqref{comm_3}, rather, we will certify the relations listed in Eqs. \eqref{comm_3} and \eqref{anticomm_3} from the maximal violation of the proposed temporal inequality.  Hence, the self-testing proof of the  measurements and the graph state $\ket{G_3}$ with the temporal non-contextuality inequality also applies to the self-testing proof using the non-contextuality inequality in \eqref{nc_new}.

Note that, we do not use the extant non-contextuality inequality from Ref. \cite{SJA22}. This is because the temporal extension of the non-contextuality inequality in \cite{SJA22} does not certify all required commutations relations.

\subsection{Self-testing with the temporal inequality}
The temporal inequality  constructed from \eqref{nc_new} has the following form
\begin{widetext}
\begin{eqnarray}\label{tncigautamjeba}
T_3&:=&\braket{A_1A_2B_3M_{12}}_{\pi}+\braket{A_3A_1B_2M_{31}}_{\pi}+\braket{A_2A_3B_1M_{23}}_{\pi}+\braket{A_1B_2B_3}_{\pi}+\braket{B_1A_2B_3}_{\pi}+\braket{B_1B_2A_3}_{\pi}\nonumber \\
&&-\braket{A_1A_2A_3}_{\pi}+\braket{M_{23}M_{31}}_{\pi}+\braket{M_{31}M_{12}}_{\pi}+\braket{M_{23}M_{12}}_{\pi}\nonumber\\ 
&\le& \eta_C=8.
\end{eqnarray}
\end{widetext}
It is direct to observe that the maximal quantum value of this expression which  
equals also the maximal algebraic value amounts to 
$\eta_Q = 10$ and is achieved by the same state and observables as the maximal quantum value of $\mathcal{I}_3$.

Our next task is to prove that the maximal quantum violation of the inequality (\ref{tncigautamjeba}) can be used for certification of the graph state (\ref{FCg}) along with the observables in \eqref{ABCs}. To this aim, consider a quantum realization given by a pure state $|\psi\ra\in\mathcal{H}$ and a set of quantum observables $A_i,B_j$ and $M_{ij}$ with $i,j\in \{1,2,3\}$  acting on $\mathcal{H}$, where $\mathcal{H}$ is some unknown Hilbert space. If the inequality in Eq. (\ref{tncigautamjeba}) is maximally violated for a pure state $\ket{\psi}\in \mathcal{H}$ then all the expectation values in the inequality, except the permutations of $\braket{A_1A_2A_3}=-1$, take the value +1. Moreover, we can also get the following set of relations
%
%
%
\begin{align}
    &A_iA_jB_kM_{ij}\ket{\psi}=\ket{\psi} \quad & \quad \text{+ permutations},  \label{3qubitallpermutationsa}\\
    &A_iB_jB_k\ket{\psi}=\ket{\psi} \quad & \quad \text{+ permutations}, \label{3qubitallpermutationsb} \\
    &A_iA_jA_k\ket{\psi}=-\ket{\psi} \quad & \quad \text{+ permutations}, \label{3qubitallpermutationsc} \\
    &M_{ij}M_{ik}\ket{\psi}=\ket{\psi} \quad & \quad \text{+ permutations} \label{3qubitallpermutationsd}, 
\end{align}
%
where $i\neq j\neq k$ and by `$+$ permutations', we mean that the relations also hold for all possible permutation of the indices. Notice that although the inequality \eqref{tncigautamjeba} does not respect the symmetries \eqref{symm1} of the non-contextuality inequality \eqref{nc_new}, the above relations \eqref{3qubitallpermutationsa}--\eqref{3qubitallpermutationsd} respect those symmetries, which is sufficient for our purpose. From the above relations, we can conclude that the following commutation relations hold on the state $\ket{\psi}$.
\begin{lemma} \label{st_comm}
Suppose the maximal quantum violation of inequality \eqref{tncigautamjeba} is observed. Then, for all $i\neq j\neq k\in \{1,2,3\}$, the operators $A_i$,$B_j$ and $M_{ij}$ satisfy the following commutation relations on the state $\ket{\psi}$, 
\begin{align}\label{Lemma1}
    [{A}_i,{A}_j]\ket{\psi}=[{A}_i,{B}_j]\ket{\psi}=[{B}_i,{B}_j]\ket{\psi}=[{A}_i,M_{ij}]\ket{\psi}\nonumber\\=[{A}_j,M_{ij}]\ket{\psi}  
    =[{B}_k,M_{ij}]\ket{\psi}=[M_{ij},M_{jk}]\ket{\psi}=0.
\end{align} 
\end{lemma}
\begin{proof}
To prove this statement we can exploit the fact that the maximal violation of inequality \eqref{tncigautamjeba} implies that the relations in \eqref{3qubitallpermutationsa}--\eqref{3qubitallpermutationsd} are satisfied.  

Let us begin with the first commutation relation in 
\eqref{Lemma1}. To prove it we observe that 
the third equality in \eqref{Lemma1} together with 
its permutations imply that 
\begin{equation}
        A_iA_j\ket{\psi}=A_jA_i\ket{\psi}=-A_k\ket{\psi},
\end{equation}
for any triple $i\neq j\neq k\in\{1,2,3\}$. Thus
\begin{equation}
    [A_i,A_j]\ket{\psi}=0\qquad (i\neq j).
\end{equation}

The remaining relations in \eqref{Lemma1} can be proven in an analogous way. Precisely,  
$[{A}_i,{B}_j]\ket{\psi}=[{B}_i,{B}_j]\ket{\psi}=0$ for any pair $i\neq j$
follow from Eq. \eqref{3qubitallpermutationsb} and its permutations, whereas the commutation relations between $A_i$ and $M_{ij}$ or $B_k$ and $M_{ij}$ follow from \eqref{3qubitallpermutationsa}. Finally, $[M_{ij},M_{jk}]\ket{\psi}=0$ is a consequence of the fourth equality in Eq. \eqref{3qubitallpermutationsd}.
%
%
\end{proof}
Next, we also prove the anti-commutation relations \eqref{anticomm_3} between the operators when acting on the state $\ket{\psi}$. 

\begin{lemma}\label{stanti_comm}
Suppose the maximal quantum violation of inequality \eqref{tncigautamjeba} is observed. Then, the operators $A_i$,$B_j$ and $M_{ij}$ along with the state $\ket{\psi}$ satisfy the following relations
\begin{align*}
    &\{A_i,B_i\}\ket{\psi} =\{A_k,M_{ij}\}\ket{\psi}\\&=\{B_i,M_{ij}\}\ket{\psi}=\{B_j,M_{ij}\}\ket{\psi}=0
\end{align*} 
for all triples $i\neq j\neq k\in \{1,2,3\}$,
\end{lemma}
\begin{proof}
As in the case of the previous lemma we can again exploit 
the relations \eqref{3qubitallpermutationsa}--\eqref{3qubitallpermutationsd}. 
Let us first observe that they lead to the following chain of relations:
\begin{eqnarray}
     A_iB_i |\psi\ra &=& A_iA_jB_k |\psi\ra = M_{ij} |\psi\ra= M_{jk} |\psi\ra \nonumber\\
     &=& B_iA_kA_j|\psi\ra  =-B_iA_i |\psi\ra,   
\end{eqnarray}
where the first equality follows from Eq. (\ref{3qubitallpermutationsb}), 
whereas the second one from Eq. (\ref{3qubitallpermutationsa}). Then, the third and fourth identities are consequences of Eqs. (\ref{3qubitallpermutationsd}) and \eqref{3qubitallpermutationsa}, respectively. The last equality follows from Eq. (\ref{3qubitallpermutationsa}). This proves that $\{A_i,B_i\}\ket{\psi}=0$ for any $i$. 

Applying the same approach, we then have 
\begin{eqnarray}
    A_iM_{jk} |\psi\ra &=& A_iM_{k,i} |\psi\ra = B_jA_k |\psi\ra= B_i |\psi\ra \nonumber\\
    &=& M_{jk}A_kA_j |\psi\ra =-M_{jk}A_i |\psi\ra,
\end{eqnarray}
and
\begin{eqnarray}
B_jM_{jk} |\psi\ra &=& B_jM_{k,i} |\psi\ra = A_iA_k |\psi\ra= -A_j |\psi\ra \nonumber\\
&=& -M_{jk}B_iA_k |\psi\ra = -M_{jk}B_j |\psi\ra,    
\end{eqnarray}
which give the remaining anticommutation relations.
A similar proof follows for $\{B_k,M_{jk}\}\ket{\psi}=0$.
\end{proof}

Lemmas \ref{st_comm} \& \ref{stanti_comm} establish thus commutation and anticommutation relations for the observables when acting on the state maximally violating our inequality. Let us notice that the identities (\ref{3qubitallpermutationsa})--(\ref{3qubitallpermutationsd}) allow us to derive even more relations for the observables and the state:
\begin{align}\label{pairwiseequalities}
    &A_iB_i|\psi\ra=-B_iA_i|\psi\ra=M_{12} |\psi\ra=M_{23} |\psi\ra=M_{31} |\psi\ra, \nonumber \\
    &A_iM_{jk}|\psi\ra=-M_{jk}A_i|\psi\ra=B_i|\psi\ra, \nonumber\\
    & B_jM_{jk}|\psi\ra=-M_{jk}B_j|\psi\ra=-A_j|\psi\ra ,\nonumber\\
    &B_kM_{jk}|\psi\ra=-M_{jk}B_k|\psi\ra=-A_k|\psi\ra ,
\end{align}
$\forall$ $i\neq j \neq k$. 
These relations will be useful for the later proofs. 

Now, inspired by the approach of Ref. \cite{IMO+20}, we define a subspace 
\beq \label{defInv}
    V_3 &:=& \mathrm{span} \{ |\psi\ra , A_1 |\psi\ra, A_2 |\psi\ra, A_3 |\psi\ra, \nonumber \\ &&\hspace{1cm}B_1 |\psi\ra, B_2 |\psi\ra, B_3 |\psi\ra, A_1 B_1 |\psi\ra \},
\eeq
and prove the following fact for it.
\begin{lemma}\label{invsubV}
Let the maximal quantum violation of inequality \eqref{tncigautamjeba} is observed, then, $V_3$ is an invariant subspace of all the observables $A_i$, $B_j$ and $M_{ij}$ for $i,j\in\{ 1,2,3 \}$.
\end{lemma}
\begin{proof}
In order to prove this lemma we use Lemmas \ref{st_comm} \& \ref{stanti_comm} as well as the relations in Eqs. \eqref{3qubitallpermutationsa}--\eqref{3qubitallpermutationsd} and \eqref{pairwiseequalities}. 

Let us begin with the first observable $A_1$ and notice that non-trivial actions of $A_1$ on the subspace elements are: 
\begin{itemize}
\item $A_1(A_2\ket{\psi})=-A_3\ket{\psi}\in V_3$, which follows from Eq. \eqref{3qubitallpermutationsc}. This relation also implies that $A_1 (A_3\ket{\psi})\in V_3$.
\item $A_1(B_1\ket{\psi})=-M_{ij}\ket{\psi}\in V_3$, which is a consequence of the first identity in Eq. \eqref{pairwiseequalities},
\item $A_1(B_2\ket{\psi})=B_3\ket{\psi}\in V_3$, which stems from Eq. (\ref{3qubitallpermutationsb}). This relation implies also that $A_1(B_3\ket{\psi})=B_2\ket{\psi}\in V_3$.
\item $A_1A_1B_1\ket{\psi}=B_1\ket{\psi} \in V_3$.
\end{itemize}
Thus, the action of $A_1$ on every vector from $V_3$ produces a vector 
belonging to $V_3$. By symmetry, the same conclusion can be drawn for $A_2$ and $A_3$. Thus the subspace is invariant under the action of all $A_i$'s. Similarly, the invariance of subspace can be shown under the action all $B_i$'s. 

Let us finally consider the observables $M_{ij}$. The first identity in Eq. \eqref{pairwiseequalities} says that $M_{ij}\ket{\psi}=A_iB_i\ket{\psi}$ for any $i$ and $j\neq i$, and therefore the action of $M_{ij}$ on $\ket{\psi}$ or $A_1B_1\ket{\psi}$ results in vectors that manifestly belong to $V_3$. It then follows from Lemmas \ref{st_comm} and \ref{stanti_comm} that $M_{ij}$ either commutes or anti-commutes with $A_i$ and $B_i$ on the state $\ket{\psi}$ which implies that the action of $M_{ij}$ keeps the remaining vectors within $V_3$ up to a negative sign. This completes the proof.
\end{proof}

It should be noticed that due to  Eq. \eqref{pairwiseequalities}, the subspace $V_3$ stays the same if one replaces the last vector $A_1B_1\ket{\psi}$ in Eq. (\ref{defInv}) by $A_2B_2 |\psi\ra$ or $A_3B_3 |\psi\ra$ or $M_{ij}\ket{\psi}$. 

An important consequence of Lemma \ref{invsubV} is that the underlying Hilbert space $\mathcal{H}$ and all the observables 
giving rise to the maximal violation of inequality \eqref{tncigautamjeba} split as $\mathcal{H}=V_3\oplus V^{\perp}_3$, where $V_3^{\perp}$ is the orthogonal complement of $V_3$ in $\mathcal{H}$, and 
\begin{equation}\label{block}
    A_i=\hat{A}_i\oplus \bar{A}_i,\quad B_{j}=\hat{B}_j\oplus \bar{B}_j, \quad M_{ij}=\hat{M}_{ij}\oplus \bar{M}_{ij},
\end{equation}
where the hatted operators $\hat{A}_i$, $\hat{B}_i$ and $\hat{M}_{ij}$
are defined on $V_3$, that is, $\hat{A}_i=PA_iP$ etc. with $P:\mathcal{H}\to V_3$ being a projection onto $V_3$, whereas the remaining ones on $V_{3}^{\perp}$.
Since $\bar{A}_i$,$\bar{B}_j$ and $\bar{M}_{ij}$
act trivially on $V_3$, that is, $\bar{A}_iV_3=\bar{B}_jV_3=\bar{M}_{ij}V_3=0$, which means that the
observed correlations giving rise to the maximal violation of 
the inequality (\ref{tncigautamjeba}) come solely from the subspace $V_3$, 
in what follows we can restrict our attention to the 
operators $\hat{A}_i$, $\hat{B}_j$ and $\hat{M}_{ij}$. 

%
%

First, from the fact that $A_i$, $B_j$ and $M_{ij}$ are observables obeying $A_i^2=B_j^2=M_{ij}^2=\mathbbm{1}$, it directly follows that $\hat{A}_i$, $\hat{B}_j$ and $\hat{M}_{ij}$ are observables too and satisfy 
\begin{equation}\label{dupa3}
 \hat{A}_i^2=\hat{B}_j^2=\hat{M}_{ij}^2=\mathbbm{1}_{V_3}\qquad (i,j=1,2,3),   
\end{equation}
where $\mathbbm{1}_{V_3}$ is the identity acting on $V_3$.
Second, Eq. (\ref{block}) implies that the hatted observables must obey the same commutation relations as $A_i$, $B_j$, and $M_{ij}$. We prove the following two lemmas in Appendix \ref{hatted_three_app}.
\begin{lemma}\label{hattedcommutations}
Suppose the maximal quantum violation of inequality \eqref{tncigautamjeba} is observed. Then, the operators $\hat{A}_i$, $\hat{B}_i$ and $\hat{M}_{ij}$ satisfy the following commutation relations $\forall i\neq j \neq k \in \{1,2,3\}$,
\begin{align*}
    [\hat{A}_i,\hat{A}_j]&=[\hat{A}_i,\hat{B}_j]=[\hat{B}_i,\hat{B}_j]=[\hat{A}_i,\hat{M}_{ij}]=[\hat{A}_j,\hat{M}_{ij}] \nonumber\\ 
    &=[\hat{B}_k,\hat{M}_{ij}]=[\hat{M}_{ij},\hat{M}_{jk}]=0.
\end{align*}
\end{lemma}
And lastly, it turns out that the anti-commutation relations in Lemma \ref{stanti_comm} also apply to the observables on the subspace $V_3$. 
\begin{lemma}\label{hattedanticommutations}
Suppose the maximal quantum violation of the inequality (\ref{tncigautamjeba}) is observed. 
Then, the following anti-commutation relations hold $\forall i\neq j \neq k \in \{1,2,3\}$,
\begin{align*}
    \{\hat{A}_i,\hat{B}_i\}=\{\hat{A}_k,\hat{M}_{ij}\}=\{\hat{B}_i,\hat{M}_{ij}\}=\{\hat{B}_j,\hat{M}_{ij}\}=0.
\end{align*}
\end{lemma}
With Lemmas \ref{hattedcommutations} and \ref{hattedanticommutations} at hand, we can now employ the standard
approach that has already been used in many non-locality-based self-testing 
schemes \cite{Kan16,KST+19,SSK+19,BAS+20}. Precisely, using this approach we can first infer that the dimension $d$ of the subspace $V_3$ is even. 
To see this, note that from the above anti-commutation 
relation between $\hat{A}_i$ and $\hat{B}_i$, we have
\begin{equation}
\hat{A}_i = - \hat{B}_i\hat{A}_i\hat{B}_i\quad\mathrm{or}\quad
\hat{B}_i = - \hat{A}_i\hat{B}_i\hat{A}_i,
\end{equation}
which after taking trace on both sides simplify to $\tr\hat{A}_i=\tr\hat{B}_i=0$. Similarly we can show that $\tr(\hat{M}_{ij})=0$. It thus follows that both the eigenvalues $\pm1$ of each observable $\hat{A}_i$, $\hat{B}_i$ or $\hat{M}_{ij}$ have equal multiplicities. This clearly implies that 
the dimension $d=\dim V_3$ is an even number, $d=2k$ for some $k\in\mathbbm{N}$, and thus $V_3=\mathbbm{C}^2\otimes\mathbbm{C}^k$. On the other hand, since $\dim V_3\leq 8$, one concludes that the possible values that $k$ can take are $k=2,3,4$.

The fact that $\hat{A}_1$, $\hat{B}_1$ and $\hat{M}_{12}$ are traceless together with the fact that they obey the commutation and anticommutation relations established in Lemmas \ref{hattedcommutations} and \ref{hattedanticommutations} imply that up to some unitary operation these operators are equivalent to $X \otimes \mathbbm{1}_k$, $Z \otimes \mathbbm{1}_k$ and $X\otimes Q_k$, where $\mathbbm{1}_k$ and $Q_k$ are the identity and some operator on $\mathbbm{C}^k$ with $k=2,3,4$ (see for instance appendix B in Ref. \cite{KST+19} for the proof of this statement). This observation is one of the key ideas behind the proof of the following lemma. 

\begin{lemma}\label{lemma3}
Suppose the maximal quantum violation of the inequality (\ref{tncigautamjeba}) is observed. Then, there exists a unitary operator $U$ acting on $V_3$ such that 
\begin{align}\label{eq:Lemma6}
U\hat{A}_i U^{\dagger}  = X_i , \quad
U\hat{B}_j U^{\dagger}= Z_j , \nonumber\\
U\hat{M}_{ij}U^{\dagger}=\pm X_i\otimes X_j \otimes Z_k.
\end{align}
\end{lemma}
\begin{proof}
Using the result (Lemma 3) of Ref. \cite{SJA22}, we can show that there exists a unitary $U$, acting on the subspace $V_3\cong\mathbbm{C}^2\otimes\mathbbm{C}^2\otimes\mathbbm{C}^2$, such that the operators $\hat{A}_i$ and $\hat{B}_i$ have the form given in Lemma \ref{lemma3}. Then, our task is to determine the form of the operators $\hat{M}_{ij}$.

Let us then move to the observables $\hat{M}_{ij}$ and consider first the $\hat{M}_{12}$. Clearly, the "rotated" matrix $U\,\hat{M}_{12}\, U^{\dagger}$ 
can be decomposed as
\begin{equation}
    U\,\hat{M}_{12}\, U^{\dagger}=\mathbbm{1}\otimes \mathsf{M}_0+
    X\otimes \mathsf{M}_1+Y\otimes \mathsf{M}_2+Z\otimes \mathsf{M}_3,
\end{equation}
where $\mathsf{M}_i$ are some $4\times 4$ matrices. Now, it is direct to observe that the conditions that $U\,\hat{M}_{12}\, U^{\dagger}$ commutes with $A_1$ and anticommutes with $B_1$ implies that 
\begin{equation}\label{formula}
    U\,\hat{M}_{12}\, U^{\dagger}=X\otimes \mathsf{M}_1.
\end{equation}
Next, the conditions that $U\,\hat{M}_{12}\, U^{\dagger}$ commutes with $A_2$ and anticommutes with $B_2$ imply that the matrix $\mathsf{M}_1$ must obey the following relations
\begin{equation}
    [X\otimes \mathbbm{1},\mathsf{M}_1]=\{Z\otimes\mathbbm{1},\mathsf{M}_1\}=0.
\end{equation}
Let us then decompose $\mathsf{M}_1$ in the following way 
\begin{equation}
    \mathsf{M}_1=\mathbbm{1}\otimes \mathsf{N}_0+
    X\otimes \mathsf{N}_1+Y\otimes \mathsf{N}_2+Z\otimes \mathsf{N}_3,
\end{equation}
where $\mathsf{N}_i$ are some $2\times 2$ matrices. After plugging
the above representation into Eq. (\ref{formula}), one easily finds that
$\mathsf{N}_0=\mathsf{N}_2=\mathsf{N}_3=0$, and therefore
\begin{equation}
    \mathsf{M}_1=X\otimes\mathsf{N}_1.
\end{equation}
To finally fix $\mathsf{N}_1$ one can exploit the fact that 
$U\,\hat{M}_{12}\, U^{\dagger}$ commutes with $B_3$ and anticommutes with $A_3$ which means that 
\begin{equation}
    [Z,\mathsf{N}_1]=\{X,\mathsf{N}_1\}=0.
\end{equation}
The only matrix compatible with the above constraints 
is $\mathsf{N}_1=\alpha Z$ for some $\alpha\in\mathbbm{C}$. 
However, since $\hat{M}_{12}$ is a quantum observable such that 
$\hat{M}_{12}^2=\mathbbm{1}$, one has that $\alpha=\pm1$. As a result, 
\begin{equation}
      U\,\hat{M}_{12}\, U^{\dagger}=\pm X_1\otimes X_2\otimes Z_3.
\end{equation}

Using then the same methodology one can then show that 
\begin{equation}
    U\,\hat{M}_{ij}\, U^{\dagger}=\pm X_i\otimes X_j\otimes Z_k
\end{equation}
for any triple $i\neq j\neq k$, and thus under the action of the unitary operation $U$, all the observables $\hat{A}_i$, $\hat{B}_j$ and $M_{ij}$ have the form given Eq. \eqref{eq:Lemma6}, which completes the proof.

%
%
\end{proof}

We can now present one of the  main results of this paper.

\begin{thm}\label{Theo3qubit}
If a quantum state $|\psi \ra$ and a set of measurements $A_i$, $B_i$ and $M_{ij}$ with $i,j\in \{1,2,3\}$ maximally violate the inequality \eqref{tncigautamjeba}, then there exists a projection $P:\mathcal{H} \rightarrow V_3$ with $V_3=(\mathbbm{C}^2)^{\otimes 3}$ and a unitary $U$ acting on $V_3$ such that
\beq\label{dupawolowa}
U^\dagger\, (P\, A_i\, P^\dagger)\, U &=& X_i, \nonumber\\
U^\dagger\, (P\, B_i\, P^\dagger)\, U &=& Z_i,\nonumber \\
U^\dagger\, (P\, M_{ij}\, P^\dagger)\, U &=& X_i\otimes X_j \otimes Z_k, 
\eeq
and the state of the form,
%
%
%
\begin{equation}\label{dupawolowa1}
    U (P|\psi \ra) = |G_3 \ra
\end{equation}
with $|G_3 \ra$ being the three-qubit complete graph state defined in Eq. (\ref{FCg}).
\end{thm}
\begin{proof}
A quantum state $\ket{\psi}$ that belongs to a Hilbert space $\mathcal{H}$ and a set of observables $\hat{A}_i$, $\hat{B}_i$, and $M_{ij}$ acting on $\mathcal{H}$ attain the maximal quantum violation of the inequality \eqref{tncigautamjeba} if and only if they satisfy the set of relations \eqref{3qubitallpermutationsa}--\eqref{3qubitallpermutationsd}. The algebraic relations induced by this set of equations let us prove Lemmas \ref{stanti_comm}-\ref{lemma3} which imply that there exists a projection $P:\mathcal{H} \rightarrow V_3 \cong \mathbbm{C}^8 $ and a unitary operation $U:\mathbbm{C}^8\to\mathbbm{C}^8$ for which Eqs. \eqref{eq:Lemma6} hold true.

From the above characterization of the observables, we can infer the form of the state $\ket{\psi}$. We will also fix the sign of all the observables $\hat{M}_{ij}$. First, after plugging $\hat{A}_i$ and
$\hat{B}_i$ into Eqs. (\ref{3qubitallpermutationsb}), one finds that 
the latter are the stabilizing relations of the three-qubit complete graph state and thus $U(P\ket{\psi})=\ket{G_3}$. One then observes that 
Eqs. (\ref{3qubitallpermutationsa}) can only be satisfied if the sign of 
each of the $U\,\hat{M}_{ij}\,U^{\dagger}$ observables in Eq. (\ref{eq:Lemma6}) is $+1$. Thus, the state and observables maximally violating our inequality are of the form (\ref{dupawolowa}-\ref{dupawolowa1}). 
%
\end{proof}

Notice that the statement in Theorem 
\ref{Theo3qubit} involves a global unitary operation, and therefore it cannot be considered only a state certification, but rather simultaneous certification of both the state and measurements giving rise to the maximal violation of our inequality.
%

%
\section{Self-testing of multi-qubit graph states}\label{scalable}
We will construct temporal non-contextual inequalities to certify multi-qubit quantum systems, without assuming any commutation relations between the operators. Importantly, the maximum violation of the new inequality will enable us to verify the required commutation relations. The new inequalities are such that the number of measurement operators grows polynomially with the number of qubits. However, the number of correlators grows rapidly with the number of qubits. The inequalities we propose generalize the inequalities given in Eqs. \eqref{nc_new} and \eqref{tncigautamjeba}.

\subsection{New scalable non-contextuality inequalities}

We present new scalable non-contextuality inequalities (analogous to the $n$-qubit inequalities in \cite{SJA22}) for which a temporal extension will enable self-testing the multi-qubit states and measurements. Consider a set of $n,n,\binom{n}{2}$ observables denoted by $A_i$, $B_i$ and $N_{ij}$ respectively, where $\{i,j\}\in \{1,\ldots,n\}$ and $i\neq j$. These operators are assumed to obey the following commutation relations
\begin{align}\label{commut_n}
 &[A_i,A_j]=[B_i,B_j]=[A_i,B_j]=[A_i,N_{ij}]=[A_j,N_{ij}]\nonumber \\ 
 &=[B_k,N_{ij}]=[N_{ij},N_{lm}]=0,
\end{align}
where $i\neq j \neq k$ and $l\neq m $. It should be noted that $A_k$ commutes with $N_{ij}$ if $k=i$ \text{or} $j$ while $B_k$ commutes with $N_{ij}$ if $k\neq i$ \text{and} $j$. Therefore, it is evident that $N_{ij}$ is the generalization of the observables $M_{ij}$ for arbitrary $n$-values. Also, note that $N_{ij}=N_{ji}$.

We will now describe our construction of scalable non-contextuality inequalities. Consider correlators of the following types,
\begin{align}\label{correlators-n}
     &\la B_1\ldots A_i\ldots A_j\ldots B_k\ldots B_nN_{ij} \ra, \quad \#\binom{n}{2} \nonumber \\
    &\la B_1\ldots B_{i-1} A_iB_{i+1}\ldots B_n \ra, \quad\quad\quad\quad \# n \nonumber \\
    &\la B_1\ldots A_i\ldots A_j \ldots A_k\ldots B_n \ra, \quad \quad\#\binom{n}{3} \nonumber \\
    &\quad\quad\quad\quad\la N_{ij}N_{jk} \ra,  \quad\quad\quad\quad\quad\quad \# 3\binom{n}{3} .
\end{align}
where $i\neq j \neq k$, and on the right-hand side we added the number of different correlators of a given type. Let us provide a few words of explanation about these correlations. The correlators of the first type in \eqref{correlators-n} are that involve two measurements $A_i$ and $A_j$ for all pairs $i\neq j$ and $n-2$ observables $B_m$ with $m\neq i\neq j$ and a single observable $N_{ij}$. The second type involves those correlators that consist of a single measurement $A_i$ and $n-1$ observables $B_j$. Then, we have correlators containing observables $A_i$, $A_j$ and $A_k$ for all triples of $i\neq j\neq k$ and $n-3$ $B_i$ observables at the remaining 'positions'.
%
%
Lastly, the correlators $\la N_{ij}N_{jk} \ra$ are self-explanatory.

Using these correlators let us now construct our non-contextuality inequality for the $n$-qubit case to be of the following form,
\begin{widetext}
\begin{align}\label{nc_new_n}
\mathcal{I}_n =& (n-2)\Big(\la A_1A_2B_3\ldots B_nN_{12} \ra+\la A_1B_2A_3\ldots B_n N_{31}\ra+\ldots  + \la B_1\ldots A_{n-1}A_nN_{n-1,n}\ra\Big)\nonumber \\ &+\alpha_n\Big(\la A_1B_2B_3B_4\ldots B_n \ra + \la B_1A_2B_3B_4\ldots B_n \ra + \la B_1B_2A_3B_4\ldots B_n \ra + \ldots  + \la B_1B_2B_3B_4\ldots A_n \ra\Big) \nonumber \\
& - \la A_1A_2A_3B_4\ldots B_n \ra - \la A_1A_2B_3A_4\ldots B_n \ra - \ldots  -
\la B_1\ldots B_{n-3}A_{n-2}A_{n-1}A_{n} \ra  +\la N_{12}N_{23}\ra\nonumber \\ &+\la N_{23}N_{31}\ra+\la N_{31}N_{12} \ra \ldots +\la N_{n-2,n-1}N_{n-1,n}\ra+\la N_{n-1,n}N_{n,n-2}\ra+\la N_{n,n-2}N_{n-2,n-1} \ra \le \eta_C^{(n)},  
\end{align}
\end{widetext}
where the constants $n-2$ and $\alpha_n=\binom{n-1}{2}$ have been added to make the numbers of correlators of different types equal (cf. inequality \eqref{nc_new}), $\eta_C^{(n)}$ stands for the maximal classical value of $\mathcal{I}_n$. To estimate the latter we can follow the pattern that one might observe after considering the expressions $\mathcal{I}_3$ in (\ref{nc_new}) and also $\mathcal{I}_4$ and $\mathcal{I}_5$ explicitly stated in Appendix \ref{4case}. Indeed, we notice that the contribution from first line of Eq. (\ref{nc_new_n}) is equal to $(n-2)\binom{n}{2}=n\alpha_n$ which is exactly equal to the contribution from second line. The last two lines contribute $(3-1)\binom{n}{3}$. Therefore, the classical upper bound of $\mathcal{I}_n$ is
\begin{equation}
    \eta_C^{(n)}=2 \alpha_n n+2\binom{n}{3}.
\end{equation}
At the same time, the maximal algebraic and at the same time the maximal quantum value of the expression $\mathcal{I}_n$ amounts to
\begin{equation}
     \eta_Q^{(n)} = 2 \alpha_n n+4\binom{n}{3}
\end{equation}
and it is strictly larger than $\eta_C^{(n)}$ for any $n$.

Note here that for $n=3$ the above inequality reproduces the one in Eq. \eqref{nc_new} by replacing $N_{ij}$ with $M_{ij}$. The inequality in Eq. \eqref{nc_new_n} is constructed in such a way so that for every three indices $i,j,k\in \{1,2,\ldots,n\}$ such that $i\neq j \neq k$, we have an inequality of the form of Eq. (\ref{nc_new}). To understand the connection, we refer enthusiastic readers to Appendix \ref{4case} where we discuss in a more detailed way the $n=4$ case.

Now, coming back to the inequality \eqref{nc_new_n}, we see that it is non-trivial for any $n$, i.e., $\eta^{(n)}_C < \eta^{(n)}_Q$. To prove this, we note that the quantum bound $\eta^{(n)}_Q$ can be attained by the following observables 
\begin{align}\label{obs_n}
&A_i = X_i, \quad
B_j = Z_j, \quad 
N_{ij}= \pm X_i\otimes X_j \bigotimes_{k\neq i,j}^n Z_k,
\end{align}
where these operators obey the commutation relations given in Eq. \eqref{commut_n}.

To use the inequality Eq. \eqref{nc_new_n} for certification of the state and measurements whenever the inequality is violated maximally, we note that the operators must also satisfy the following relations,
\begin{align}\label{nqubitallpermutations}
    & B_1\ldots A_i...\ldots A_j\ldots B_k\ldots B_nN_{ij} \ket{\psi}=\ket{\psi}\nonumber \\
    & B_1\ldots B_{i-1} A_iB_{i+1}\ldots B_n \ket{\psi}=\ket{\psi}\nonumber \\
    & B_1\ldots A_i\ldots A_j \ldots A_k\ldots B_n \ket{\psi}=-\ket{\psi}\nonumber \\
    & N_{ij}N_{jk} \ket{\psi}=\ket{\psi},
\end{align}
where these relations hold true for all possible permutations of operators. In the next section, we will use the above relations to obtain the form of operators $A_i, B_j$ and $N_{ij}$ and the state. In fact, we will obtain the compatibility relations and  the relations in \eqref{nqubitallpermutations} from the maximum violation of the temporal version of the inequality \eqref{tncigautamjeba_n}. 

Also, note here that, we do not use the extant $n$-qubit non-contextuality inequality from Ref. \cite{SJA22} as its temporal extension does not certify all required commutations relations.

\subsection{Self-testing with temporal extension of n-qubit non-contextuality inequality}
Using the similar prescription for $n=3$, we construct the temporal version of the $n$-qubit non-contextuality inequality \eqref{nc_new_n}, by replacing the sequential expectation value terms with a sum of sequential correlations \eqref{nth}. The exact form of our $n$-qubit temporal inequality is given by 
\begin{widetext}
\begin{align}\label{tncigautamjeba_n}
    T_n =& (n-2)\Big(\la A_1A_2B_3\ldots B_nN_{12} \ra_{\pi}+\la A_1B_2A_3\ldots B_n N_{31}\ra_{\pi} + \ldots  +\la B_1\ldots A_{n-1}A_nN_{n-1,n}\ra_{\pi}\Big)\nonumber \\ &+\alpha_n\Big(\la A_1B_2B_3B_4\ldots B_n \ra_{\pi} + \la B_1A_2B_3B_4\ldots B_n \ra_{\pi} + \la B_1B_2A_3B_4\ldots B_n \ra_{\pi} + \ldots  + \la B_1B_2B_3B_4\ldots A_n \ra_{\pi}\Big) \nonumber \\
& - \la A_1A_2A_3B_4\ldots B_n \ra_{\pi} - \la A_1A_2B_3A_4\ldots B_n \ra_{\pi} - \ldots  -
\la B_1\ldots B_{n-3}A_{n-2}A_{n-1}A_{n} \ra_{\pi}\nonumber\\
&+\la N_{12}N_{23}\ra_{\pi}+\la N_{23}N_{31}\ra_{\pi}+\la N_{31}N_{12} \ra_{\pi} \ldots +\la N_{n-2,n-1}N_{n-1,n}\ra_{\pi}+\la N_{n-1,n}N_{n,n-2}\ra_{\pi}+\la N_{n,n-2}N_{n-2,n-1} \ra_{\pi} \nonumber \\ &\le \eta_C^{(n)}=2 \alpha_n n+2\binom{n}{3} < \eta_Q^{(n)} = 2 \alpha_n n+4\binom{n}{3}.
\end{align}
\end{widetext}

We will now use the above inequality for certification of the complete graph states of any number of qubits, and the measurements $A_i$ ,$B_i$ and $N_{ij}$ given in Eq. (\ref{obs_n}). 
From the maximum violation of the above inequality, we have all the temporal expectation value terms in $T_n$ taking the value +1, except the expectation value terms of the form $\la B_1\ldots  A_i\ldots A_j \ldots A_k \ldots B_n \ra_{\pi}$ which takes the value -1. In other words, the maximal violations give us the relations in \eqref{nqubitallpermutations}. Using this we can prove the following lemma,
\begin{lemma} \label{st_comm_n}
Suppose the maximal quantum violation of the inequality (\ref{tncigautamjeba_n}) is observed. Then, the operators $A_i$, $B_i$ and $N_{ij}$ obey the following commutation relations,
\begin{align}
[A_i,A_j]\ket{\psi}=[A_i,B_j]\ket{\psi}=[B_i,B_j]\ket{\psi}=[N_{ij},N_{jk}]\ket{\psi}\nonumber \\ =[A_i,N_{ij}]\ket{\psi}=[A_j,N_{ij}]\ket{\psi}=[B_k,N_{ij}]\ket{\psi}=0,\nonumber
\end{align}
where $i\neq j \neq k$.
\end{lemma}
\begin{proof}The proof is analogous to the one of Lemma 
Firstly, we define $\mathcal{B}_{ijk}=\prod_{\substack{m\neq i,j,k}}^n B_{m}$ and then, by using the relations in \eqref{nqubitallpermutations}, we can get the following relations which give the above commutators;  
\begin{align*}
    \mathcal{B}_{ijk}A_kA_iA_j\ket{\psi}=\mathcal{B}_{ijk}A_kA_jA_i\ket{\psi}=-\ket{\psi},\\
    \mathcal{B}_{ijk}B_kA_iB_j\ket{\psi}=\mathcal{B}_{ijk}B_kB_jA_i\ket{\psi}=\ket{\psi},\\
    \mathcal{B}_{ijk}A_kB_iB_j\ket{\psi}=\mathcal{B}_{ijk}A_kB_jB_i\ket{\psi}=\ket{\psi} ,\\
    \mathcal{B}_{ijk}A_jB_kA_iN_{ij}\ket{\psi}=\mathcal{B}_{ijk}A_jB_kN_{ij}A_i\ket{\psi}=\ket{\psi}, \\
    \mathcal{B}_{ijk} A_iA_jB_kN_{ij}\ket{\psi}=\mathcal{B}_{ijk} A_iA_jN_{ij}B_k\ket{\psi}=\ket{\psi},\\
    {\rm and}\quad N_{ij}N_{jk}\ket{\psi}=N_{jk}N_{ij}\ket{\psi}=\ket{\psi}.
\end{align*}
Note that the proof of $[A_i,N_{ij}]=0$ is the same as the proof of $[A_j,N_{ij}]=0$ because of the symmetry of $N_{ij}$ with respect to its indices.
\end{proof}
Moreover, we also have anti-commutation relations in the form of the following lemma.
\begin{lemma}\label{stanti_comm_n}
Suppose the maximal quantum violation of the inequality (\ref{tncigautamjeba_n}) is observed. Then, the operators $A_i$, $B_i$ and $N_{ij}$ obey the following anti-commutation relations,
\begin{align*}
\{A_i,B_i\} |\psi\ra=\{A_i,N_{jk}\} |\psi\ra=\{B_j,N_{jk}\} |\psi\ra \\=\{B_k,N_{jk}\} |\psi\ra=0,\:\: {\rm where}\:\: i\neq j \neq k.
\end{align*} 
\end{lemma}
\begin{proof}
Again, we will be using the relations from \eqref{nqubitallpermutations} for the proof. We have the following relations, 
\begin{align*}
    A_iB_i |\psi\ra = A_iB_kA_j\mathcal{B}_{ijk} |\psi\ra = N_{ij} |\psi\ra= N_{jk} |\psi\ra \\
    =B_iA_kA_j\mathcal{B}_{ijk}|\psi\ra=
 -B_iA_i |\psi\ra,\\
 A_iN_{jk} |\psi\ra = A_iN_{i,k} |\psi\ra = A_kB_j\mathcal{B}_{ijk} |\psi\ra 
 = B_i |\psi\ra \\=N_{jk}A_kA_j \mathcal{B}_{ijk}|\psi\ra=-N_{jk}A_i |\psi\ra,\\
 B_jN_{jk} |\psi\ra = B_jN_{i,k} |\psi\ra = A_kA_i \mathcal{B}_{ijk}|\psi\ra 
 = -A_j |\psi\ra  \\=-N_{jk}A_kB_i \mathcal{B}_{ijk}|\psi\ra=-N_{jk}B_j |\psi\ra.
\end{align*}
 which gives the required anti-commutation relations. Note that the proof for $\{B_j,N_{jk}\} |\psi\ra=0$ is similar to $\{B_k,N_{jk}\} |\psi\ra=0$.
\end{proof}
Lemmas \ref{st_comm_n} and \ref{stanti_comm_n} thus establish commutation and anticommutation relations between observables on the state maximally violating our inequality. What is more, together with equations they enable and the relations in  \eqref{nqubitallpermutations}, we can also claim the following relations,
\begin{align}\label{pairwiseequalities_n}
    &A_iB_i|\psi\ra=-B_iA_i|\psi\ra=A_jB_j|\psi\ra=-B_jA_j|\psi\ra\nonumber =N_{ij} |\psi\ra\\ &\quad =N_{jk} |\psi\ra=N_{i,k} |\psi\ra=N_{k,l} |\psi\ra, \:\:\forall i\neq j \neq k \neq l.\nonumber \\
    &A_iN_{jk}|\psi\ra=-N_{jk}A_i|\psi\ra=B_i|\psi\ra, \hspace{2mm}\forall i\neq j \neq k.\nonumber\\
    &B_jN_{jk}|\psi\ra=-N_{jk}B_j|\psi\ra=-A_j|\psi\ra, \hspace{2mm}\forall j \neq k.\nonumber\\
    &B_kN_{jk}|\psi\ra=-N_{jk}B_k|\psi\ra=-A_k|\psi\ra, \hspace{2mm}\forall j \neq k.
\end{align}
Now we will use these relations to show that the following subspace is invariant under the action of the operators $A_i$, $B_i$ and $N_{ij}$
\begin{eqnarray}
    V_{n}&=&\mathrm{span}\{\ket{\psi},B_{i_1}\ket{\psi},B_{i_1}B_{i_2}\ket{\psi},\ldots, B_{i_1}B_{i_2}\ldots B_{i_{k}}\ket{\psi},\nonumber\\
    &&\hspace{0.3cm}\ldots,B_{i_1}B_{i_2}\ldots B_{i_{n-1}}\ket{\psi},B_{1}\ldots B_{n}\ket{\psi}\},
\end{eqnarray}
where $i_j \in [1,n]$ for any $j$ and $i_1<\ldots <i_{k-1}<i_k<\ldots <i_{n-1}$. In the simplest cases of $n=3$ and $n=4$, the above construction gives 
\begin{align*}
    V_3=\mathrm{span}\{\ket{\psi},B_i\ket{\psi},B_{i}B_{j}\ket{\psi},B_1B_2B_3\ket{\psi}\},\:\: {\rm and}\\
     V_4=\mathrm{span}\{\ket{\psi},B_i\ket{\psi},B_iB_j\ket{\psi},B_iB_jB_k\ket{\psi},\prod_{i=1}^4 B_i\ket{\psi}\},
\end{align*}
with $i\neq j\neq k\in [1,4]$. In particular, 
$V_3$ is exactly the same as the one defined in 
Eq. \eqref{defInv} by noticing that due to Eq. \eqref{3qubitallpermutationsb},
$B_iB_j\ket{\psi}=A_k\ket{\psi}$ with $i\neq j\neq k$
and $B_1B_2B_3\ket{\psi}=-A_1B_1\ket{\psi}$.

This is the same subspace as was found in \cite{IMO+20} for the $n$-qubit scenario. Hence, the number of vectors spanning $V_n$ is $2^n$ which implies $\dim V_n\leq 2^n$. In fact, we show later that $\dim V_n=2^n$.

Our aim now is to identify the form of the operators $A_i$, $B_j$, and $N_{ij}$ projected onto the subspace $V_n$. The method of the proof of self-testing is similar to that used in the case $n=3$. In what follows, we prove the following lemma in Appendix \ref{nqubit_app},

\begin{lemma}\label{invsubVn}
Suppose the maximal quantum violation of the inequality (\ref{tncigautamjeba_n}) is observed. Then, the subspace $V_n$ of $\mathcal{H}_n$ is invariant under the action of the operators  $A_i$, $B_j$ and $N_{ij}$ for $i,j=[1,n ]$.
\end{lemma}
%


Now, lemma \ref{invsubVn} implies that $A_i$, $B_j$ and $N_{ij}$ can be represented as a direct sum of two blocks, 
\begin{equation}
A_i = \hat{A}_i \oplus \bar{A}_{i},\qquad  B_j = \hat{B}_j \oplus \bar{B}_{j} \qquad N_{ij}=\hat{N}_{ij}\oplus \bar{N}_{ij},\nonumber
\end{equation}
where $\hat{A}_i$, $\hat{B}_i$ and $\hat{N}_{ij}$ are projections of $A_i$ , $B_i$ and $N_{ij}$ onto $V_n$, 
that is $\hat{A}_i=P_n A_i P_n$, $\hat{B}_i=P_n B_i P_n$ and $\hat{N}_{ij}=P_n N_{ij} P_n$ with $P_n:\mathcal{H}_n\to V_n$
denoting the projector onto $V_n$. On the other hand, $\bar{A}_{i}$, $\bar{B}_{j}$ and $\bar{N}_{ij}$
are defined on the orthogonal complement of $V_n$ in the Hilbert space $\mathcal{H}_n$ that we denote
$V_n^{\perp}$; clearly, $\mathcal{H}_n=V_n\oplus V_n^{\perp}$. 

Importantly, $\bar{A}_i$, $\bar B_i$ and $\bar{N}_{ij}$ act trivially on the subspace $V_n$, in particular $\bar{A}_i\ket{\psi}=B_i'\ket{\psi}=\bar{N}_{ij}\ket{\psi}=0$, and consequently
it is enough for our purposes  to characterize $\hat{A}_i$, $\hat{B}_j$ and $\hat{N}_{ij}$. Our first step to achieving this goal is to show that the following commutation and anti-commutation relations between the hatted operators hold. We prove these relations in Appendix \ref{nqubit_app}. 
\begin{lemma}\label{hattedcommutations_n}
Suppose the maximal quantum violation of the inequality (\ref{tncigautamjeba}) is observed, then the following holds, 
\begin{align*}
    [\hat{A}_i,\hat{A}_j]=[\hat{A}_i,\hat{B}_j]=[\hat{B}_i,\hat{B}_j]=[\hat{B}_k,\hat{N}_{ij}]=[\hat{A}_i,\hat{N}_{ij}]\\=[\hat{A}_j,\hat{N}_{ij}]=[\hat{N}_{ij},\hat{N}_{jk}]=0,\:\: \forall i\neq j\neq k\in [ 1, n].
\end{align*}
\end{lemma}
\begin{lemma}\label{hattedanticommutations_n}
Suppose the maximal quantum violation of the inequality (\ref{tncigautamjeba}) is observed, then the bellow relations holds $\forall$ $i\neq j\neq k\in [1, n]$,
\begin{align*}
    \{\hat{A}_i,\hat{B}_i\}=\{\hat{A}_k,\hat{N}_{ij}\}=\{\hat{B}_i,\hat{N}_{ij}\}=\{\hat{B}_j,\hat{N}_{ij}\}=0.
\end{align*}
\end{lemma}


Next, we are in a position to state the following results in a lemma which is a straightforward generalization of Lemma \ref{lemma3}.
\begin{lemma}\label{lemma3n}
Suppose the maximal quantum violation of our inequality (\ref{tncigautamjeba_n}) is observed.
Then, there exists a unitary matrix $U$ acting on $V_n$ for which
\begin{align}\label{Aitilden}
    U\hat{A}_i U^{\dagger} = X_i, \quad U\hat{B}_jU^{\dagger} = Z_j, \nonumber\\ U\hat{N}_{ij}U^{\dagger}= \pm X_i\otimes X_j \bigotimes_{k\neq i,j}^n Z_k.
\end{align} 
\end{lemma}
\begin{proof}
Again, using the result (Lemma 6) of Ref. \cite{SJA22}, we can show that there exists a unitary $U$, acting on the subspace $V_n\cong(\mathbbm{C}^2)^{\otimes n}$, such that the operators $\hat{A}_i$ and $\hat{B}_i$ have the form given in Lemma \ref{lemma3n}. Then, our task is to determine the form of the operators $\hat{N}_{ij}$.

%
%
%
Using the structure of subspace, $V_n=(\mathbbm{C}^2)^{\otimes n}$, one can decompose $\hat{N}_{ij}$ in terms of the Pauli basis as
\begin{equation*}
   U\hat{N}_{ij}U^{\dagger}=\bigotimes_{m=1}^n\sigma_m,
\end{equation*}
where $\sigma_m$ is a linear combination of Pauli matrices $\{\mathbbm{1}_m,X_m,Y_m,Z_m\}$. Now, using the same technique used in the proof of Lemma \ref{lemma3}, we get the following relations from Lemmas \ref{hattedcommutations_n} and \ref{hattedanticommutations_n},
\begin{align*}
    & \text{for}  \quad m=i \quad \sigma_m=X, \nonumber \\
    & \text{for} \quad m=j \quad \sigma_m=X, \nonumber \\
    & \text{for} \quad m\neq i,j \quad \sigma_m=\pm Z.
\end{align*}
which gives the form of $\hat{N}_{ij}$.
\end{proof}

Thus, we are able to certify the measurements $A_i$, $B_i$, and $N_{ij}$, without assuming commutation relations between the observables. Moreover, as we can show that the dimension of $V_n$ is exactly $2^n$, our inequalities can be seen as dimension witnesses: 
\textit{The dimension of the Hilbert space supporting a state and observables giving rise to the maximal violation of our inequalities must be at least $2^n$.} Now, we present the main result in the form of the following theorem. 

\begin{thm}\label{nqubit_theorem}
If a quantum state $|\psi \ra$ and a set of dichotomic observables $A_i$, $B_j$ and $N_{ij}$ with $i,j\in [1,n]$ give rise to the maximal violation of the temporal non-contextual inequality \eqref{tncigautamjeba_n}, then there exists a projection $P_n:\mathcal{H}_n \rightarrow V_n$ with $V_n=\mathbbm{C}^{2^n}$ and a unitary $U$ acting on $\mathbbm{C}^{2^n}$ such that
\beq
U^\dagger (P\, A_i\, P^\dagger) U &=& X_i, \nonumber\\
U^\dagger (P\, B_j\, P^\dagger) U &=& Z_j, \nonumber\\
U^\dagger (P\, N_{ij}\, P^\dagger) U &=& X_i\otimes X_j \bigotimes_{k\neq i,j}^n Z_k, \nonumber\\
U (P|\psi \ra) &=& |G_n \ra,
\eeq
where $|G_n \ra$ is the complete graph state of $n$ qubits. 
\end{thm}
\begin{proof}
The state $|\psi \ra\in\mathcal{H}_n$ and observables $A_i,B_i,N_{ij}$ acting on the Hilbert space $\mathcal{H}_n$ attain the maximal quantum violation of the inequality \eqref{tncigautamjeba_n}, if and only if, they satisfy the set of $\binom{n}{2}+n+\binom{n}{3}+3\binom{n}{3}$ equations in \eqref{nqubitallpermutations}. The algebra induced by this set of equations allows us to prove Lemmas \ref{invsubVn}, \ref{hattedcommutations_n}, \ref{hattedanticommutations_n} and \ref{lemma3n}; in particular, it follows that there exists a projection $P_n:\mathcal{H} \rightarrow V_n \cong \mathbbm{C}^{2^n} $ and a unitary $U$ acting on $V_n$ such that the relations \eqref{Aitilden} holds. Then, we note that the first equation of \eqref{nqubitallpermutations} is satisfied only if the the sign of each observables, $ U\hat{N}_{ij}U^{\dagger}$ in \eqref{Aitilden} is $+1$. 
%
%

Now, we note that the $n$ stabilizing relations from the second equation of \eqref{nqubitallpermutations} uniquely determine a state, $\ket{G_n}$, associated with the complete graph with $n$ vertices. 
%
%
Therefore, we prove the theorem.
\end{proof}
%
%
\section{Conclusions}
Quantum contextuality provides a unique avenue in extending the task of self-testing quantum devices to scenarios where entanglement is not necessary or spatial separation between the
subsystems is not required (cf., \cite{IMO+20,SSA20,MMJ+21,SJA22,BRV+19,DMS+22}). However, such contextuality-based certification schemes rely on the compatibility relations among the measurements involved and are thus generally hard to implement in practice.

In this paper, we overcome this limitation by using the temporal inequalities derived from the non-contextuality inequalities, such that the maximal violation of these new inequalities will certify these commutation relations among the said measurement observables. In what follows, we are able to show that the maximum violation of these temporal non-contextual inequalities can be used for the certification of multi-qubit graph states and the measurements. 

Recently, we developed a similar method to certify a two-qubit maximally entangled state and measurements from the maximal violation of a temporal non-contextual inequality \cite{twoqubit_temp}. Moreover, we were also able to show that the scheme presented in \cite{twoqubit_temp} is also robust to noise and small experimental errors. Since the certification schemes presented in both works are similar, we expect that the certification scheme presented in the current work is also robust to noise and small experimental errors. 

Furthermore, it would be interesting to identify which other non-contextuality inequality-based self-testing schemes can be converted to temporal inequalities so that they can still be used for self-testing of the states and measurements. Further research is needed to improve the scalability of our scheme with the number of certified qubits. 
\section*{acknowledgments}
This work was supported by the Polish National Science Centre through the SONATA BIS project No. 2019/34/E/ST2/00369. SS acknowledges funding through PASIFIC program call 2 (Agreement No. PAN.BFB.S.BDN.460.022 with the Polish Academy of Sciences). This project has received funding from the European Union’s Horizon 2020 research and innovation programme under the Marie Skłodowska-Curie grant agreement No 847639 and from the Ministry of Education and Science of Poland.

\bibliography{mermin}
\onecolumngrid
\appendix

%
%
\section{The special case of $n=4$ of non-contextual inequality \eqref{nc_new_n}}\label{4case}
The non-contextual inequality \eqref{nc_new_n} is constructed in such a way that for every index $i\neq j\neq k\in [1,n]$, we have an inequality of the form of Eq. \eqref{nc_new}. We explain it by expanding the inequality for $n=4$, 
\begin{align} \label{nc_new_4}
\mathcal{I}_4  = &2( \la A_1A_2B_3B_4N_{12}\ra+\la A_1A_3B_2B_4N_{13}\ra+\la A_1A_4B_2B_3N_{14} \ra+\la A_2A_3B_1B_4N_{23} \ra+\la A_2A_4B_1B_3N_{24} \ra \nonumber \\ &+\la A_3A_4B_1B_2N_{34}\ra) +3\left(\la A_1B_2B_3B_4 \ra + \la B_1A_2B_3B_4 \ra + \la B_1B_2A_3B_4 \ra + \la B_1B_2B_3A_4 \ra \right)\nonumber \\ &- \la B_1A_2A_3A_4 \ra - \la A_1B_2A_3A_4 \ra - \la A_1A_2B_3A_4 \ra - \la A_1A_2A_3B_4 \ra+\la N_{12}N_{13}\ra+\la N_{23}N_{13}\ra \nonumber \\
&+\la N_{12}N_{23}\ra+\la N_{12}N_{14}\ra+\la N_{24}N_{14}\ra+\la N_{12}N_{24}\ra+\la N_{14}N_{13}\ra+\la N_{34}N_{13}\ra+\la N_{14}N_{34}\ra\ \nonumber \\ 
&+\la N_{24}N_{34}\ra+\la N_{23}N_{34}\ra+\la N_{24}N_{23}\ra \leq \eta^{(4)}_C = 32 < \eta^{(4)}_Q = 40.
\end{align}
Rearranging the terms in $\mathcal{I}_4$ so that for the four sets of three indices, we have one of the four expressions,
\begin{align}
\mathcal{I}_4&=(\braket{B_4A_1A_2B_3N_{12}}+\braket{B_4A_1B_2A_3N_{31}}+\braket{B_4B_1A_2A_3N_{23}} +\la B_4A_1B_2B_3 \ra + \la B_4B_1A_2B_3 \ra + \la B_4B_1B_2A_3 \ra \nonumber \\ & -\la B_4A_1A_2A_3\ra +\la N_{12}N_{13}\ra+\la N_{23}N_{13}\ra+\la N_{12}N_{23}\ra)\nonumber\\
&+(\braket{B_3A_1A_2B_4N_{12}}+\braket{B_3A_1B_2A_4N_{14}}+\braket{B_3B_1A_2A_4N_{24}} +\la B_3A_1B_2B_4 \ra + \la B_3B_1A_2B_4 \ra + \la B_3B_1B_2A_4 \ra \nonumber \\ & -\la B_3A_1A_2A_4\ra +\la N_{12}N_{14}\ra+\la N_{24}N_{14}\ra+\la N_{12}N_{24}\ra)\nonumber\\
&+(\braket{B_2A_1B_3A_4N_{14}}+\braket{B_2A_1A_3B_4N_{13}}+\braket{B_2B_1A_3A_4N_{34}} +\la B_2A_1B_3B_4 \ra + \la B_2B_1B_3A_4 \ra + \la B_2B_1A_3B_4 \ra \nonumber \\ & -\la B_2A_1A_3A_4\ra +\la N_{14}N_{13}\ra+\la N_{34}N_{13}\ra+\la N_{14}N_{34}\ra)\nonumber\\
&+(\braket{B_1A_2B_3A_4N_{24}}+\braket{B_1A_2A_3B_4N_{23}}+\braket{B_1B_2A_3A_4N_{34}} +\la B_1A_2B_3B_4 \ra + \la B_1B_2B_3A_4 \ra + \la B_1B_2A_3B_4 \ra \nonumber \\ & -\la B_1A_2A_3A_4\ra +\la N_{24}N_{23}\ra+\la N_{34}N_{23}\ra+\la N_{24}N_{34}\ra).
\end{align}
From the above form, it is clear that $\mathcal{I}_4$ consists of four expressions of the form \eqref{nc_new}. For example, when $i=1,j=2$ and $k=4$ in Eqs. \eqref{correlators-n}, we have the following terms for $\mathcal{I}_4$
\begin{align*}
\braket{B_3A_1A_2B_4N_{12}}&+\braket{B_3A_1B_2A_4N_{14}}  +\braket{B_3B_1A_2A_4N_{24}} +\la B_3A_1B_2B_4 \ra + \la B_3B_1A_2B_4 \ra \\&+ \la B_3B_1B_2A_4 \ra  -\la B_3A_1A_2A_4\ra  +\la N_{12}N_{14}\ra+\la N_{24}N_{14}\ra+\la N_{12}N_{24}\ra,
\end{align*}
where we see that we have a $B_3$ sitting in from of all the terms from inequality \eqref{nc_new} except for $\la N_{ij}N_{jk}\ra$ terms. Likewise, we have three other expressions. Therefore, it is direct to find that $\eta^{(4)}_C = 4\times 8= 32$.
\section{Commutation and anti-commutation relations for hatted operators in three-qubit scenario}\label{hatted_three_app}
\textit{The proof of Lemma \ref{hattedcommutations}}.-- 
If we can show that the operators $A_i$, $B_j$, and $M_{ij}$ follow the aforementioned commutation relations on the subspace $V_3$, it would also immediately imply the commutation relations for the hatted operators. Now consider the commutators $[A,B]$ acting on $M\ket{\psi}\in V_3$ so that $A,B,M$ appear together in any sequential correlation terms in inequality \eqref{tncigautamjeba}. We can always get the following 
$[A,B]M\ket{\psi}=M[A,B]\ket{\psi}$ 
from the relations \eqref{3qubitallpermutationsa}-\eqref{3qubitallpermutationsd} which contain all the possible permutations. So in the following, we exclude such trivial scenarios and will consider only those commutators and subspace vectors so that the operators $A, B, M$ do not appear together in any sequential correlation terms in inequality \eqref{tncigautamjeba}. 

For the commutation $[A_i,A_j]=0 \in V_3$, we find that the following relation is the only non-trivial one to show, 
\begin{align*}
    [A_i,A_j]B_i\ket{\psi}=(A_iB_k-A_jM_{ij})\ket{\psi}=0
\end{align*}
which follows from $A_jB_iB_k\ket{\psi}=\ket{\psi}$ from Eq. \eqref{3qubitallpermutationsb}, $B_kA_iA_jM_{ij}\ket{\psi}=\ket{\psi}$ (Eq. \ref{3qubitallpermutationsa}) and $A_iB_i\ket{\psi}=M_{ij}\ket{\psi}$ (Eq. \ref{pairwiseequalities}). Similarly, for $[A_i,B_j]=0\in V_3$, we need to show the following two relations only. First one is 
\begin{align*}
    [A_i,B_j]B_i\ket{\psi}=(A_iA_k-B_jM_{ik})\ket{\psi}=0
\end{align*}
using $A_kB_jB_i\ket{\psi}=\ket{\psi}$ from Eq. \eqref{3qubitallpermutationsb}, $A_kA_iB_jM_{ik}\ket{\psi}=\ket{\psi}$ from Eq. \eqref{3qubitallpermutationsa} and $A_iB_i\ket{\psi}=M_{jk}\ket{\psi}$ from Eq. \eqref{pairwiseequalities}. And the second one is found to be 
\begin{align*}
    [A_i,B_j]A_j\ket{\psi}=(-A_iM_{ik}+B_jA_k)\ket{\psi}=0
\end{align*}
using $B_jA_j\ket{\psi}=-M_{ik}\ket{\psi}$ from Eq. \eqref{pairwiseequalities}, $A_kA_iA_j\ket{\psi}=-\ket{\psi}$ from Eq. \eqref{3qubitallpermutationsc} and $M_{ik}A_iB_jA_k\ket{\psi}=\ket{\psi}$ from Eq. \eqref{3qubitallpermutationsa}. Next, for $[B_i,B_j]=0\in V_3$, we need to prove three non-trivial cases. We first have 
\begin{align*}
    [B_i,B_j]A_i\ket{\psi}=(B_iB_k-B_jM_{ik})\ket{\psi}=(A_j+A_iA_k)\ket{\psi}=0
\end{align*}
using $B_kB_jA_i\ket{\psi}=\ket{\psi}$ from Eq. \eqref{3qubitallpermutationsb}, $A_jA_iA_k\ket{\psi}=-\ket{\psi}$ from Eq. \eqref{3qubitallpermutationsc} and $B_iA_i\ket{\psi}=-M_{ik}$ from Eq. \eqref{pairwiseequalities}. It is easy to see that a similar proof works for $[B_i,B_j]A_j\ket{\psi}=0$. Further, we can see that following holds 
\begin{align*}
    [B_i,B_j]B_k\ket{\psi}=(B_iA_i-B_jA_j)\ket{\psi}=0
\end{align*}
by using $A_iB_jB_k\ket{\psi}=\ket{\psi}$ and $A_jB_iB_k\ket{\psi}=\ket{\psi}$ from Eq. \eqref{3qubitallpermutationsb}. Lastly, we find below one  
\begin{align*}
    [B_i,B_j]M_{ij}\ket{\psi}=(-B_iA_j+B_jA_i)\ket{\psi}=(B_k-B_k)\ket{\psi}=0
\end{align*}
using $B_jM_{ij}\ket{\psi}=-A_j\ket{\psi}$ from Eq. \eqref{3qubitallpermutationsa}, $A_iB_jB_k\ket{\psi}=\ket{\psi}$ from Eq. \eqref{3qubitallpermutationsb} and $B_iM_{ij}\ket{\psi}=-A_i\ket{\psi}$ from Eq. \eqref{pairwiseequalities}. 

For $[A_i,M_{ij}]=0\in V_3$, we show the following relation 
\begin{align*}
    [A_i,M_{ij}]A_k\ket{\psi}=(-A_iB_k+M_{ij}A_j)\ket{\psi}=0
\end{align*}
using $M_{ij}A_k\ket{\psi}=B_k\ket{\psi}$ from Eq. \eqref{pairwiseequalities} and $A_jA_iA_k\ket{\psi}=-\ket{\psi}$ and $B_kA_iM_{ij}A_j\ket{\psi}=\ket{\psi}$ from Eq. \eqref{3qubitallpermutations}. Next, 
\begin{align*}
    [A_i,M_{ij}]B_i\ket{\psi}=(A_iA_i-M_{ij}M_{ij})\ket{\psi}=0,
\end{align*}
by considering  $M_{ij}B_i\ket{\psi}=A_i\ket{\psi}$ and $A_iB_i\ket{\psi}=M_{ij}\ket{\psi}$ from Eq. \eqref{pairwiseequalities}. Lastly, we have 
$[A_i,M_{ij}]B_j\ket{\psi}=(A_iA_j-M_{ij}B_k)\ket{\psi}=0$ which follows a similar arguments. It is also clear that the above proofs work equally for $[A_j,M_{ij}]=0 \in V_3$, as the operator $M_{ij}$ is symmetric with respect to the exchange of indices. Now, for $[B_k,M_{ij}]=0\in V_3$, we show 
\begin{align*}
    [B_k,M_{ij}]A_k\ket{\psi}=(-B_kB_k+M_{ij}M_{ij})\ket{\psi}=0
\end{align*}
holds using $M_{ij}A_k\ket{\psi}=-B_k\ket{\psi}$ and $B_kA_k\ket{\psi}=M_{ij}\ket{\psi}$ from Eq. \eqref{pairwiseequalities}). Then, one can prove following
\begin{align*}
    [B_k,M_{ij}]B_i\ket{\psi}=(B_kA_i-M_{ij}A_j)\ket{\psi}=0
\end{align*}
by utilizing $M_{ij}B_i\ket{\psi}=A_i\ket{\psi}$ from Eq. \eqref{pairwiseequalities}, $A_jB_kB_i\ket{\psi}=\ket{\psi}$ and $A_iB_kM_{ij}A_j\ket{\psi}=\ket{\psi}$ from Eq. \eqref{3qubitallpermutations}.

For $[M_{ij},M_{jk}]=0\in V_3$, we have the following four non-trivial commutation relations. First one being
\begin{align*}
    [M_{ij},M_{jk}]A_i\ket{\psi}=(-M_{ij}B_i-M_{jk}A_jB_k)\ket{\psi}=(-A_i-M_{jk}B_i)\ket{\psi}=(-A_i-A_jA_k)\ket{\psi}=0
\end{align*}
using $M_{jk}A_i\ket{\psi}=-B_i\ket{\psi}$, $M_{ij}B_i\ket{\psi}=A_i$ from Eq. \eqref{pairwiseequalities}, $M_{ij}A_iA_jB_k\ket{\psi}=\ket{\psi}$,  $B_iA_jB_k\ket{\psi}=\ket{\psi}$, $A_jA_kB_iM_{jk}\ket{\psi}=\ket{\psi}$, and $A_iA_jA_k\ket{\psi}=-\ket{\psi}$ from Eq. \eqref{3qubitallpermutationsa}-\eqref{3qubitallpermutationsc}. Second, we show 
\begin{align*}
    [M_{ij},M_{jk}]A_j\ket{\psi}=(M_{ij}A_kB_i-M_{jk}A_iB_k)\ket{\psi})=(M_{ij}B_j-M_{jk}B_j)\ket{\psi}=0
\end{align*}
using $B_iA_kM_{jk}A_j\ket{\psi}=\ket{\psi}$, $B_kA_iM_{ij}A_j\ket{\psi}=\ket{\psi}$, $A_kB_i\ket{\psi}=A_iB_k\ket{\psi}=B_j\ket{\psi}$ from Eq. \eqref{3qubitallpermutationsa}-\eqref{3qubitallpermutationsc} and $M_{ij}B_j\ket{\psi}=A_j\ket{\psi}$ from Eq. \eqref{pairwiseequalities}. Third, one can deduce that the below relation
\begin{align*}
    [M_{ij},M_{jk}]B_i\ket{\psi}=(M_{ij}A_jA_k-M_{jk}A_i)\ket{\psi}=(-M_{ij}A_i+B_i)\ket{\psi}=(-A_jB_k+B_i)\ket{\psi}=0
\end{align*}
does hold as $A_kA_jM_{jk}B_i\ket{\psi}=\ket{\psi}$, $M_{ij}B_i\ket{\psi}=A_i\ket{\psi}$, $A_iA_jA_k\ket{\psi}=-\ket{\psi}$, $B_kA_jM_{ij}A_i\ket{\psi}=\ket{\psi}$, $B_iA_jB_k\ket{\psi}=\ket{\psi}$ from Eq. \eqref{3qubitallpermutationsa}-\eqref{3qubitallpermutationsc} and $M_{jk}A_i\ket{\psi}=-B_i\ket{\psi}$ from Eq. \eqref{pairwiseequalities}. Finally, it is easy to show that 
\begin{align*}
    [M_{ij},M_{jk}]B_j\ket{\psi}=(M_{ij}A_j-M_{jk}A_j)\ket{\psi}=(A_iB_k-A_kB_i)\ket{\psi}=(B_j-B_j)\ket{\psi}=0
\end{align*}
holds using $M_{jk}B_j\ket{\psi}=A_j\ket{\psi}$, $M_{ij}B_j\ket{\psi}=A_j\ket{\psi}$ from Eq. \eqref{pairwiseequalities}, $B_kA_iM_{ij}A_j\ket{\psi}=\ket{\psi}$ plus its permutations, $B_jA_iB_k\ket{\psi}=\ket{\psi}$ and $B_jA_kB_i\ket{\psi}=\ket{\psi}$ from Eqs. \eqref{3qubitallpermutationsa}-\eqref{3qubitallpermutationsb}. $\qed$

\vspace{5mm}
\textit{The proof of Lemma \ref{hattedanticommutations}}.--
As we had proven for the commutators, if we can show that the anti-commutators from Lemma \ref{stanti_comm} hold on the subspace $V_3$, then the anti-commutation relations will also hold for the hatted operators. For $\{A_i,B_i\}=0\in V_3$, we need to consider five non-trial cases. The first one,
\begin{align*}
    \{A_i,B_i\}A_i\ket{\psi}=(-A_iM_{jk}+B_i)\ket{\psi}=0
\end{align*}
follows using $B_iA_i\ket{\psi}=-M_{jk}\ket{\psi}$ and $A_iM_{jk}\ket{\psi}=B_i\ket{\psi}$ from Eq. \eqref{pairwiseequalities}. Then, one shows  
\begin{align*}
    \{A_i,B_i\}A_j\ket{\psi}=(A_iB_k-B_iA_k)\ket{\psi}=0
\end{align*}
using $B_kB_iA_j\ket{\psi}=\ket{\psi}$, $A_kA_iA_j\ket{\psi}=-\ket{\psi}$, $B_jA_iB_k\ket{\psi}=B_jB_iA_k\ket{\psi}=\ket{\psi}$ from Eq. \eqref{3qubitallpermutationsb}-\eqref{3qubitallpermutationsc}. Next, we find that  
\begin{align*}
    \{A_i,B_i\}B_i\ket{\psi}=(A_i+B_iM_{ij})\ket{\psi}=0
\end{align*}
holds using the relation $A_iB_i\ket{\psi}=M_{ij}\ket{\psi}$ and $B_iM_{ij}\ket{\psi}=-A_i\ket{\psi}$ from Eq. \eqref{pairwiseequalities}. Further, we find the following relation 
\begin{align*}
    \{A_i,B_i\}B_j\ket{\psi}=(A_iA_k+B_iB_k)\ket{\psi}=0
\end{align*}
using $B_iB_j\ket{\psi}=A_k\ket{\psi}$, $B_kA_iB_j\ket{\psi}=\ket{\psi}$, $A_jA_iA_k\ket{\psi}=-\ket{\psi}$, and $A_jB_iB_k\ket{\psi}=\ket{\psi}$ from Eq. \eqref{3qubitallpermutationsb}-\eqref{3qubitallpermutationsc}; and lastly, 
\begin{align*}
    \{A_i,B_i\}A_iB_i\ket{\psi}=(-A_iB_iB_iA_i+B_iA_iA_iB_i)\ket{\psi}=0
\end{align*}
follows using $A_iB_i\ket{\psi}=-B_iA_i\ket{\psi}$ from Eq. \eqref{pairwiseequalities}. Then for $\{A_k,M_{ij}\}=0\in V_3$, we need to show the following relations. We can prove following 
\begin{align*}
    \{A_k,M_{ij}\}A_k\ket{\psi}=(-A_kB_k+M_{ij})\ket{\psi}=0
\end{align*}
using $M_{ij}A_k\ket{\psi}=-B_k\ket{\psi}$ and $A_kB_k\ket{\psi}=M_{ij}\ket{\psi}$ from Eq. \eqref{pairwiseequalities}. Next one is  
\begin{align*}
    \{A_k,M_{ij}\}A_i\ket{\psi}=(A_kA_jB_k-M_{ij}A_j)=(A_kB_i-A_iB_k)\ket{\psi}=0
\end{align*}
using $B_kA_jM_{ij}A_i\ket{\psi}=\ket{\psi}$, $A_jA_kA_i\ket{\psi}=-\ket{\psi}$, $B_iA_jB_k\ket{\psi}=\ket{\psi}$, $B_kA_iM_{ij}A_j\ket{\psi}=\ket{\psi}$, and $B_jA_kB_i\ket{\psi}=B_jA_iB_k\ket{\psi}=\ket{\psi}$ from Eq. \eqref{3qubitallpermutationsa}-\eqref{3qubitallpermutationsc}. Then, we show the following holds
\begin{align*}
    \{A_k,M_{ij}\}B_k\ket{\psi}=(A_kA_jA_i+M_{ij}M_{ij})\ket{\psi}=0
\end{align*}
utilizing $A_iA_jM_{ij}B_k\ket{\psi}=\ket{\psi}$ from Eq. \eqref{3qubitallpermutationsa}, $A_kA_jA_i\ket{\psi}=-\ket{\psi}$ from Eq. \eqref{3qubitallpermutationsc} and $A_kB_k\ket{\psi}=M_{ij}$ from Eq. \eqref{pairwiseequalities}. Further, we show
\begin{align*}
    \{A_k,M_{ij}\}B_i\ket{\psi}=(A_kA_i+M_{ij}B_j)\ket{\psi}=(-A_j+A_j)\ket{\psi}=0
\end{align*}
using $M_{ij}B_i\ket{\psi}=A_i\ket{\psi}$, $M_{ij}B_j\ket{\psi}=A_j\ket{\psi}$ from Eq. \eqref{pairwiseequalities}, $B_jA_kB_i\ket{\psi}=\ket{\psi}$ from Eq. \eqref{3qubitallpermutationsb}, $A_jA_iA_k\ket{\psi}=-\ket{\psi}$ from Eq. \eqref{3qubitallpermutationsc}; and lastly, 
\begin{align*}
    \{A_k,M_{ij}\}M_{ij}\ket{\psi}=(A_k+M_{ij}B_k)\ket{\psi}=(-A_iA_j+A_iA_j)\ket{\psi}=0
\end{align*}
using $A_kM_{ij}\ket{\psi}=B_k\ket{\psi}$ from Eq. \eqref{pairwiseequalities}, $A_jA_iA_k\ket{\psi}=-\ket{\psi}$ from Eq. \eqref{3qubitallpermutationsc} and $A_jA_iM_{ij}B_k\ket{\psi}=\ket{\psi}$ from Eq. \eqref{3qubitallpermutationsa}. Next for $\{B_i,M_{ij}\}=0\in V_3$, we can show 
\begin{align*}
    \{B_i,M_{ij}\}A_i\ket{\psi}=(B_iA_jB_k-M_{ij}M_{jk})\ket{\psi}=0
\end{align*}
by using $B_kA_jM_{ij}A_i \ket{\psi}=\ket{\psi}$ and $B_iA_i\ket{\psi}=-M_{ij}\ket{\psi}$ from Eq. \eqref{pairwiseequalities}, $B_jA_iB_k\ket{\psi}=\ket{\psi}$ from Eq. \eqref{3qubitallpermutationsb} and $M_{ij}M_{jk}\ket{\psi}=\ket{\psi}$ from Eq. \eqref{3qubitallpermutationsd}; 
\begin{align*}
    \{B_i,M_{ij}\}A_j\ket{\psi}=(B_iA_iB_k+M_{ij}B_k)\ket{\psi}=(B_iB_j+A_iA_j)\ket{\psi}=0
\end{align*}
using $B_kA_iM_{ij}A_j\ket{\psi}=\ket{\psi}$, $B_kB_iA_j\ket{\psi}=\ket{\psi}$, $B_jA_iB_k\ket{\psi}=\ket{\psi}$, $A_jA_iM_{ij}B_k\ket{\psi}=\ket{\psi}$, $A_kB_iB_j\ket{\psi}=\ket{\psi}$, and $A_kA_iA_j\ket{\psi}=-\ket{\psi}$ from Eq. \eqref{3qubitallpermutationsa}-\eqref{3qubitallpermutationsc}; 
\begin{align*}
    \{B_i,M_{ij}\}A_k\ket{\psi}=(-B_iB_k+M_{ij}B_j)\ket{\psi}=(-A_j+A_j)\ket{\psi}=0
\end{align*}
utilizing $M_{ij}A_k\ket{\psi}=-B_k\ket{\psi}$, $M_{ij}B_j\ket{\psi}=A_j\ket{\psi}$ from Eq. \eqref{pairwiseequalities}, $B_jB_iA_k\ket{\psi}=\ket{\psi}$, and $A_jB_iB_k\ket{\psi}=\ket{\psi}$ from Eq. \eqref{3qubitallpermutationsb}; 
\begin{align*}
    \{B_i,M_{ij}\}B_i\ket{\psi}=(B_iA_i+M_{ij})\ket{\psi}=0
\end{align*}
using $M_{ij}B_i\ket{\psi}=A_i\ket{\psi}$ and $B_iA_i\ket{\psi}=-M_{ij}\ket{\psi}$ from Eq. \eqref{pairwiseequalities}; 
\begin{align*}
    \{B_i,M_{ij}\}B_j\ket{\psi}=(B_iA_j+M_{ij}A_k)\ket{\psi}=(B_k-B_k)\ket{\psi}=0
\end{align*}
by using $M_{ij}B_j\ket{\psi}=A_j\ket{\psi}$, $M_{ij}A_k\ket{\psi}=-B_k\ket{\psi}$ from Eq. \eqref{pairwiseequalities}, $A_kB_iB_j\ket{\psi}=\ket{\psi}$ and $B_kB_iA_j\ket{\psi}=\ket{\psi}$  from Eq. \eqref{3qubitallpermutationsb}; 
\begin{align*}
    \{B_i,M_{ij}\}B_k\ket{\psi}=(B_iA_iA_j+M_{ij}A_j)\ket{\psi}=(-B_iA_k+A_iB_k)\ket{\psi}=0
\end{align*}
utilizing $A_jA_iM_{ij}B_k\ket{\psi}=\ket{\psi}$, $A_jB_iB_k\ket{\psi}=\ket{\psi}$, $A_kA_iA_j\ket{\psi}=\ket{\psi}$, $B_kA_iM_{ij}A_j\ket{\psi}=\ket{\psi}$, $B_jB_iA_k\ket{\psi}=\ket{\psi}$, and $B_jA_iB_k\ket{\psi}=\ket{\psi}$ from Eq. \eqref{3qubitallpermutationsa}-\eqref{3qubitallpermutationsc}; and lastly, 
\begin{align*}
    \{B_i,M_{ij}\}M_{ij}\ket{\psi}=(B_i+M_{ij}A_i)\ket{\psi}=(B_i+A_jB_k)\ket{\psi}=0
\end{align*}
using $B_iM_{ij}\ket{\psi}=A_i\ket{\psi}$ from Eq. \eqref{pairwiseequalities}, $B_kA_jM_{ij}A_i\ket{\psi}=\ket{\psi}$ and $B_iA_jB_k\ket{\psi}=\ket{\psi}$ from Eq. \eqref{3qubitallpermutationsa}-\eqref{3qubitallpermutationsb}.

Finally, we notice that a proof for $\{B_j,M_{ij}\}=0$ $\in V_3$ holds similarly to $\{B_i,M_{ij}\}=0$. $\qed$

\section{Invariant subspace, commutation and anti-commutation relations for hatted operators in $n$-qubit scenario}\label{nqubit_app}
{\it The proof of Lemma \ref{invsubVn}}.-- We have the following by using $A_1B_2\ldots B_n\ket{\psi}=\ket{\psi}$ from Eq. \eqref{nqubitallpermutations} and $B_iA_i\ket{\psi}=-N_{ij}\ket{\psi}$ from Eq. \eqref{pairwiseequalities_n},
\begin{align}\label{nelementsvector}
    B_1\ldots B_n\ket{\psi}=B_1A_1\ket{\psi}=-N_{ij}\ket{\psi}.
\end{align}
Then, by using $A_k\prod_{i_j\neq k}^n B_{i_j}\ket{\psi}=\ket{\psi}$ from Eq. \eqref{pairwiseequalities_n}, we also have the following relations,
\begin{align}\label{n-1elements}
    \prod_{i_j\neq k}^n B_{i_j}\ket{\psi}=A_k\ket{\psi}.
\end{align}
These two identities are crucial for the proof of the lemma. We will consider the different cases separately as follows.

\textbf{Action of $B_m$ on $V_n$}:
Let us first consider the following different cases when $B_m$ acts on on $B_{i_1}B_{i_2}\ldots B_{i_{k}}\ket{\psi}$. Using relation $B_1\ldots B_{i-1} A_iB_{i+1}\ldots B_n \ket{\psi}=\ket{\psi}$ from Eq. \eqref{nqubitallpermutations}, we find
\begin{align*}
    B_mB_{i_1}B_{i_2}\ldots B_{i_{k}}\ket{\psi}=(B_mB_{i_j})B_{i_1}B_{i_2}\ldots B_{i_{k}}\ket{\psi},
\end{align*}
where $m=i_j$ for $j\in [1,k]$ and $k\in [1,n-1]$. Then we notice the following hold for $k=n$, 
\begin{align*}
B_mB_1\ldots B_n \ket{\psi}=-B_mN_{jm}\ket{\psi}= A_m\ket{\psi}=\prod_{i_j\neq m}^n B_{i_j}\ket{\psi},
\end{align*}
where the first equality comes from Eq. \eqref{nelementsvector}. The second and last equality comes from $B_mN_{jm}\ket{\psi}=-A_m\ket{\psi}$ (Eq. \ref{pairwiseequalities_n}) and $A_m\prod_{i_j=1,i_j\neq m}^nB_{i_j}\ket{\psi}=\ket{\psi}$ (Eq. \ref{nqubitallpermutations}) respectively. Thus, $B_m$ keeps the state in $V_n$ for $m=i_j$. Now, for the cases with $m\neq i_j$ where $j\in [1,k]$, first, we consider $k\in [1,n-2]$, and find that the following hold,
\begin{align*}
    B_mB_{i_1}B_{i_2}\ldots B_{i_{k}}\ket{\psi}=B_{i_1}B_{i_2}\ldots B_{i_{j-1}}B_mB_{i_{j+1}}\ldots B_{i_{k}}\ket{\psi} \quad {\rm for} \quad i_{j-1}< m< i_{j+1},
\end{align*}
where we used the fact that $B_{i_1}B_{i_2}\ldots B_{i_{j-1}}B_mB_{i_{j+1}}\ldots B_{i_{k}}\ket{\psi}$ is equivalent to $B_mB_{i_1}B_{i_2}\ldots B_{i_{k}}\ket{\psi}$ under permutation from Eq. \eqref{nqubitallpermutations}. The last case for $k=n-1$, we show
\begin{align*}
    B_m\prod_{i_j\neq m}^n B_{i_j}\ket{\psi}=B_mA_m\ket{\psi}=B_1A_1\ket{\psi}=B_1\ldots B_n\ket{\psi},
\end{align*}
where we get the first equality from \eqref{n-1elements}, then used $B_mA_m\ket{\psi}=B_1A_1\ket{\psi}=N_{ij}\ket{\psi}$ from Eq. \eqref{pairwiseequalities_n} and $A_1B_2\ldots B_n\ket{\psi}=\ket{\psi}$ from Eq. \eqref{nqubitallpermutations} for immediate steps.  Thus the subspace $V_n$ is invariant under the action of $B_m$ operators for any $m\in [1,n]$.

\textbf{Action of $A_m$ on $V_n$}: We consider the action of $A_m$ operators on $B_{i_1}B_{i_2}\ldots B_{i_{k}}\ket{\psi}$ in different cases. For the cases with $m=i_j$ where $j\in [1,k]$ and $k\in [1,n-1]$, we find 
\begin{align*}
    A_mB_{i_1}B_{i_2}\ldots B_{i_j}\ldots B_{i_{k}}\ket{\psi}=A_mA_{j'}\mathcal{B}_{j'}\ket{\psi},
\end{align*}
where $\prod_{i_{j'}\neq j',m,i_j}^{n}B_{i_{j'}}=\mathcal{B}_{j'}$. We get above relation using permutations of $B_1\ldots B_{i-1} A_iB_{i+1}\ldots B_n \ket{\psi}=\ket{\psi}$ from Eq. \eqref{nqubitallpermutations} by moving the operators $A_{j'}$ and $A_m$ towards the right. Then, for $k=n$, we get 
\begin{align*}
A_mB_1\ldots B_n \ket{\psi}=-A_mN_{jk}\ket{\psi}=
-A_mA_mB_m\ket{\psi} =-B_m\ket{\psi},
\end{align*}
where the first equality follows from Eq. \eqref{nelementsvector} and last one using $N_{jk}\ket{\psi}=A_mB_m\ket{\psi}$ from Eq. \eqref{pairwiseequalities_n}. Lastly, for $m\neq i_j$ with $j\in [1,k]$ and $k\in [1,n-1]$, we show the following
\begin{align*}
    A_mB_{i_1}B_{i_2}\ldots B_{i_{k}}\ket{\psi}=\prod_{i_{j'}\neq m,i_j}^{n}B_{i_{j'}}\ket{\psi},
\end{align*}
where we used the various permutations of $A_m B_{i_1}\ldots B_{i_n}\ket{\psi}=\ket{\psi}$ from Eq. \eqref{pairwiseequalities_n}. Therefore $A_m$ keeps $V_n$ invariant.

\textbf{Action of $N_{lm}$ on $V_n$}: Finally, we will consider the action of $N_{lm}$ operators on $B_{i_1}B_{i_2}\ldots B_{i_{k}}\ket{\psi}$ in different cases. Let us assume $l,m \neq i_j$ with $j\in [1,k]$ and $k\in \{1,n-2\}$, then we have 
\begin{align*}
    N_{lm}B_{i_1}B_{i_2}\ldots B_{i_{k}}\ket{\psi}=A_{j'}A_m\mathcal{B}_{j'}\ket{\psi},
\end{align*}
using the permutations of the identity $A_mA_{j'}N_{lm}\mathcal{B}_{j'}\ket{\psi}=\ket{\psi}$ from Eq. \eqref{nqubitallpermutations}. And, whenever $(l=i_j ,m\neq i_{j})$ or $(m=i_j ,l\neq i_{j})$ for $k=[1,n-2]$, we find 
\begin{align*}
    N_{lm}B_{i_1}\ldots B_{i_{k}}\ket{\psi}=N_{lm}B_{i_{1'}}\ldots B_{i_{k'}} A_m\ket{\psi}=N_{lm}N_{lm}B_{i_1}\ldots B_{i_{k}} A_l\ket{\psi},
\end{align*}
where using the permutations of the identity $\prod_{i_j\neq m}^nB_{i_j}A_m\ket{\psi}=\ket{\psi}$ s.t., $i_{j'}\neq i_j$, and the identity $N_{lm}A_mA_l\prod_{i_j\neq m,l}^nB_{i_j}\ket{\psi}=\ket{\psi}$ from Eq. \eqref{nqubitallpermutations}. Now, the scenario when both $l$ and  $m$ are equal to one of the $i_{j}$'s with $k=\{1,n-2\}$, we find 
\begin{align*}
    N_{lm}B_{i_1}\ldots B_{i_{k}}\ket{\psi}=N_{lm}B_{i_{1'}}\ldots B_{i_{k'}}A_{i_{k'+1}}\ket{\psi}=-N_{lm}\prod_{i_j\neq l,m,k'+1,i_{j'}}^nB_{i_j}A_lA_m\ket{\psi}=-N_{lm}N_{lm}\mathcal{B}_{j'}\ket{\psi}.
\end{align*}
using the permutations of the identity $\prod_{i_s\neq k'+1}^nB_{i_{s}}A_{k'+1}\ket{\psi}=\ket{\psi}$, then $\prod_{i_s\neq k'+1,l,m}^nB_{i_{s}}A_{k'+1}A_lA_m\ket{\psi}=-\ket{\psi}$ from Eq. \eqref{nqubitallpermutations} and the permutations of  $N_{lm}\prod_{i_s\neq l,m}^nB_{i_{s}}A_lA_m\ket{\psi}=\ket{\psi}$ from Eq. \eqref{pairwiseequalities_n}. Next, we get from Eq. \eqref{n-1elements}  
\begin{align*}
    N_{lm}A_k\ket{\psi}=-B_k\ket{\psi}, \quad k\neq l,m; \:\:{\rm and}\:\:
    N_{lm}A_l\ket{\psi}=A_m\prod_{i_j \neq l,m}^n B_{i_j} \ket{\psi},
\end{align*}
by using the identity from Eq. \eqref{pairwiseequalities_n} for the first equality and the permutations of $N_{lm}A_mA_l\prod_{i_j\neq m,l}^nB_{i_j}\ket{\psi}=\ket{\psi}$ from Eq. \eqref{nqubitallpermutations} for the second one. Lastly, using Eq. \eqref{pairwiseequalities_n} in Eq. \eqref{nelementsvector} we find
$N_{lm}N_{ij}\ket{\psi}=\ket{\psi}$, where $i,j$ can be the same or different from $l,m$. Thus, $N_{lm}$ keeps $V_n$ invariant. $\qed$

\vspace{2mm}
{\it The proof of Lemma \ref{hattedcommutations_n}}.--
Similar to the proof of Lemma \ref{hattedcommutations}, we can show that the commutators of unhatted operators vanish on all the subspace elements. Again if the operators $A,B,N$ appear together in inequality \eqref{tncigautamjeba_n}, seeing $[A,B]N\ket{\psi}=N[A,B]\ket{\psi}$ holds, we ignore such scenarios in the $n$-qubit case also. We consider $[A_l,A_m]$ acting on $V_n$, for a demonstration, the remaining commutators can be proven similar to the Lemma \ref{hattedcommutations}.

For $(l=i_j,m\neq i_j)$ or $(m=i_j,l\neq i_j)$ with $j \in [1,k]$ and $k\leq n-1$, we find using $A_m\prod_{i_j\neq m}^nB_{i_j}\ket{\psi}=\ket{\psi}$ from Eq. \eqref{nqubitallpermutations},
    \begin{align*}
        &[A_l,A_m]\prod_{i_j=1}^n B_{i_j}\ket{\psi}=\left(A_l\prod_{i_{j'}\neq l,m,i_j}^nB_{i_{j'}}-A_mA_lA_m\prod_{i_{j'}\neq i_j,l,m}^nB_{i_{j'}}\right)\ket{\psi} =\left(A_l\prod_{i_{j'}\neq l,m,i_j}^nB_{i_{j'}}-A_mN_{lm}\prod_{i_j\neq l,m}^nB_{i_j}\right)\ket{\psi},
    \end{align*}
where the second equality comes from $A_lA_mN_{lm}\prod_{i_j\neq l,m}^nB_{i_j}\ket{\psi}=\ket{\psi}$ using Eq. \eqref{nqubitallpermutations}. When both $l,m$ are equal to some $i_j$ we have 
\begin{align*}
    &[A_l,A_m]\prod_{i_j=1}^n B_{i_j}\ket{\psi}=\left(A_lA_mA_{m'}\prod_{i_{j'}\neq l,m,m',i_j}^nB_{i_{j'}}-A_mA_lA_{m'}\prod_{i_{j'}\neq l,m,m',i_j}^nB_{i_{j'}}\right)\ket{\psi}=0,
\end{align*}
using $A_{m'}\prod_{i_j \neq m'}^n B_{i_j}\ket{\psi}=\ket{\psi}$, and the permutations of $A_lA_mA_{m'}\prod_{i_j\neq l,m, m'}^nB_{i_j}\ket{\psi}=-\ket{\psi}$ from Eq. \eqref{nqubitallpermutations}. 
Therefore, $[\hat A_l,\hat A_m]=0$ $\in V_n$.  $\qed$




\vspace{3mm}
{\em The proof of Lemma \ref{hattedanticommutations_n}}.-- 
For the anti-commutators, we will show the proof for $\{A_i,B_i\}$ for a demonstration, and the remaining follows similar to Lemma \ref{hattedanticommutations}.  

Using $A_j\ket{\psi}$ from Eq. \eqref{n-1elements} for $i=j$, we have $\{A_i,B_i\}A_i\ket{\psi}=(A_iN_{jk}+B_i)\ket{\psi}=0$ by using $B_iA_i\ket{\psi}=-N_{jk}\ket{\psi}$ from Eq. \eqref{pairwiseequalities_n} twice. Similarly, for $i\neq j$, we have 
\begin{align*}
     \{A_i,B_i\}A_j\ket{\psi}=\left(A_iB_k\prod_{i_j\neq i,j,k}^n B_{i_j}-B_iA_k\prod_{i_j\neq i,j,k}^n B_{i_j}\right)\ket{\psi}=0,
\end{align*}
using $\prod_{i_j\neq i,j,k}^nB_{i_j}B_iA_jB_k\ket{\psi}=\ket{\psi}$, $\prod_{i_j\neq i,j,k}^n B_{i_j}A_iA_jA_k\ket{\psi}=-\ket{\psi}$ and $\prod_{i_j\neq i,j,k}^nB_{i_j}A_iB_kB_j\ket{\psi}=\prod_{i_j\neq i,j,k}^nB_{i_j}B_iA_kB_j\ket{\psi}=\ket{\psi}$ from Eq. \eqref{nqubitallpermutations}. Then, using Eq. \eqref{nelementsvector} and $A_iB_i\ket{\psi}=-B_iA_i\ket{\psi}$ from Eq. \eqref{pairwiseequalities_n}, we have
\begin{align*}
     -\{A_i,B_i\}B_iA_i\ket{\psi}=(-\ket{\psi}-B_iA_iB_iA_i)\ket{\psi}=(-\ket{\psi}+\ket{\psi})=0.
\end{align*}
Note that the following two cases don't appear in the case of Lemma \ref{hattedanticommutations}, and therefore, we consider them to complete the proof. Considering a state of the form $\prod_{j=1}^kB_{i_j}\ket{\psi}$ with $k\in [1,n-2]$, we have 
\begin{align*}
    \{A_i,B_i\}\prod_{i_j=1}^nB_{i_j}\ket{\psi}=\left(A_i\prod_{i_j\neq i}^n B_{i_j}+B_iA_iA_k\prod_{i_{j'}\neq i,k,i_j}^n B_{i_{j'}}\right)\ket{\psi}=\left(\prod_{i_{j'}\neq i,l,i_j}^n B_{i_{j'}}-B_iA_l\prod_{i_{j}\neq i,k,i_{j'}}^n B_{i_{j}}\right)\ket{\psi}=0,
\end{align*}
where $i=i_j$ and we use the permutations of $A_k\prod_{i_j\neq k}^nB_{i_j}\ket{\psi}=\ket{\psi}$, $A_i\prod_{i_j\neq i}^nB_{i_j}\ket{\psi}=\ket{\psi}$, $A_iA_kA_l\prod_{i_j\neq i,k,l}^nB_{i_j}\ket{\psi}=-\ket{\psi}$ and $A_l\prod_{i_j\neq l}^nB_{i_j}\ket{\psi}=\ket{\psi}$ from Eq. \eqref{nqubitallpermutations}. Lastly, we find the below relation on the same state with $i\neq i_j$ 
\begin{align*}
    \{A_i,B_i\}\prod_{i_j=1}^nB_{i_j}\ket{\psi}=\left(A_iA_k\prod_{i_{j'}\neq i,k,i_j}^n B_{i_{j'}}+\prod_{i_{j'}\neq i_j}^n B_{i_{j'}}\right)\ket{\psi}=\left(-A_l\prod_{i_{j}\neq i,k,l,i_{j'}}^n B_{i_{j}}+\prod_{i_{j'}\neq i_j}^n B_{i_{j'}}\right)\ket{\psi}=0,
\end{align*}
using the fact that $A_k\prod_{i_j\neq k}^nB_{i_j}\ket{\psi}=\ket{\psi}$, $A_i\prod_{i_j\neq i}^nB_{i_j}\ket{\psi}=\ket{\psi}$, the permutations of $A_iA_kA_l\prod_{i_j\neq i,k,l}^nB_{i_j}\ket{\psi}=-\ket{\psi}$, and $A_l\prod_{i_j\neq l}^nB_{i_j}\ket{\psi}=\ket{\psi}$ from Eq. \eqref{nqubitallpermutations}.
$\qed$

\end{document}